\documentclass{article}

\usepackage{algorithm, algpseudocode}

\usepackage[english]{babel}

\usepackage[letterpaper,top=2cm,bottom=2cm,left=3cm,right=3cm,marginparwidth=1.75cm]{geometry}

\usepackage{physics,bbm,mathtools,amsmath,amsthm,amssymb,amsfonts}
\usepackage{graphicx}
\usepackage[colorlinks=true, allcolors=blue]{hyperref}
\usepackage{appendix}

\usepackage[font=small,labelfont=bf]{caption}
\usepackage{subcaption}
\usepackage{float}

\def\bu{{\mathbf{u}}}
\def\i{\mathrm{i}}
\def\e{\mathrm{e}}

\newcommand{\APF}{\operatorname{APF}}
\newcommand{\interspace}{\vspace*{0.5cm}} 

\newtheorem{proposition}{Proposition}[section]
\newtheorem{lemma}{Lemma}[section]
\newtheorem{corollary}{Corollary}[section]

\newcommand{\eqlabelApp}[1]{\stepcounter{equation}\tag{A-{\theequation}} \label{#1}}

\let\hat\widehat
\let\tilde\widetilde

\newcommand\blfootnote[1]{
    \begingroup
    \renewcommand\thefootnote{}\footnote{#1}
    \addtocounter{footnote}{-1}
    \endgroup
}

\title{Filtering through a topological lens: homology for point processes on the time-frequency plane}

\author{J.M. Miramont\footnote{Corresponding authors: J. M. Miramont, K. A. Tan, S. S. Mukherjee \\ \hspace*{0.25cm} E-mails: jmiramontt[at]univ-lille[dot]fr; kinaun[dot]tan[at]u[dot]nus[dot]edu;   ssmukherjee[at]isical[dot]ac[dot]in} \footnote{Universit\'e de Lille, CNRS, Centrale Lille, UMR9189–CRIStAL, Av. Paul Langevin,
59650 Villeneuve-d'Ascq, France.}, K.A. Tan$^{*}$\footnote{Department of Mathematics, National University of Singapore, 10 Lower Kent Ridge Road,
Singapore 119076.}, S.S. Mukherjee$^{*}$\footnote{Statistics and Mathematics Unit,
Indian Statistical Institute,
203 Barrackpore Trunk Road,
Kolkata, India 700108.}, R. Bardenet$^{\dag}$, S. Ghosh$^{\ddag}$
\blfootnote{\underline{Keywords}: Topological data analysis $|$ Time-frequency analysis $|$ Persistent homology $|$ Spectrogram zeros $|$ Stochastic geometry $|$ Hyperuniformity $|$ Point processes} 
}

\date{ }

\graphicspath{{figures/}}
\begin{document}

\maketitle

\begin{abstract}
We introduce a very general approach to the analysis of signals from their noisy measurements from the perspective of Topological Data Analysis (TDA). While TDA has emerged as a powerful  analytical tool for data with pronounced topological structures, here we demonstrate its applicability for general problems of signal processing, without any a-priori geometric feature. Our methods are well-suited to a wide array of time-dependent signals in different scientific domains, with acoustic signals being a particularly important application. We invoke time-frequency representations of such signals, focusing on their zeros which are  gaining salience as a signal processing tool in view of their stability properties. Leveraging state-of-the-art topological concepts, such as stable and minimal volumes,  we develop a complete suite of TDA-based methods to explore the delicate stochastic geometry of these zeros, capturing signals based on the disruption they cause to this rigid, hyperuniform spatial structure. Unlike classical spatial data tools, TDA is able to capture the full spectrum of the stochastic geometry of the zeros, thereby leading to powerful inferential outcomes that are underpinned by a principled statistical foundation. This is reflected in the power and versatility of our applications, which include competitive performance in processing a wide variety of audio signals (esp. in low SNR regimes), effective detection \textit{and reconstruction} of gravitational wave signals (a reputed signal processing challenge with non-Gaussian noise), and medical time series data from EEGs, indicating a wide horizon for the approach and methods introduced in this paper. 
\end{abstract}

\section{Introduction}

The interplay of randomness and structure is a defining characteristic of modern day data science and machine learning. The size and dimensionality of data we can access and store is vast and ever-growing, potentially overwhelming existing computational capacities. This has motivated the quest for simple, low-dimensional structure in the randomness of massive datasets, which is an effective ansatz to capture how the real world manifests itself in big data.
\\

\subsection{Topological Data Analysis.} An effective tool that has emerged in recent years to identify and analyze such structure is \textit{Topological Data Analysis} (\textit{abbrv.} TDA). Leveraging ideas and techniques from algebraic topology, this suite of methods aims to understand latent geometric structures that are often present in large-scale datasets. By developing a principled multi-scale approach to understand the topological structure of data point clouds, TDA isolates geometric features (such as holes) that \textit{persist} across a large number of scales (a technique referred to as \textit{persistent homology}). These persistent features may be identified to be inherent or intrinsic to the dataset that encapsulates its latent geometric structure, which in turn reveals the fundamental real-world information hidden in the ambient noise and many confounding attributes. 

Originating in the late 20th century, TDA has quickly emerged to demonstrate practical utility in a wide range of applications, starting from materials science \cite{NMTFYOHMK24,CiHiVi23} and genomics \cite{NiLeCa11} to fingerprinting proteins \cite{xia2014persistent} and audio signals \cite{RFDHB24}, biomedicine \cite{iqbal2021classification}, neuroscience \cite{colombo2022tool} and spatial networks \cite{biscio2019accumulated, byrne2019topological}, to provide a very limited and partial list of successful use cases. However, while TDA provides an attractive combination of powerful and sophisticated mathematical tools brought to the service of real-world data, its many successful applications have mostly been focused on situations where the data is fundamentally \textit{geometric} in its nature, or at least has an underlying notion of \textit{shape} in a reasonably straightforward sense \cite{kerber2016persistent, carlsson2020topological}. However, most real-world data generating processes do not embody a clear geometric (or shape) structure in any direct manner. This naturally leads to the tantalizing question of bringing the formidable arsenal of TDA techniques to bear on wider classes of data analytical problems without an a priori geometric structure. 

In this work, we focus on a piece of this puzzle, proposing a principled approach towards the large-scale application of TDA methods to a very general class of signal processing problems. In particular, our approach is readily applicable to a broad bandwidth of time-indexed signals, including as a special case all kinds of noisy signals arising in acoustics. We show the range of applicability extends much farther, encompassing gravitational wave signals (that are known for their notorious difficulty due to extremely low SNR and non-Gaussianity of noise), and general classes of time-series data (such as medical data from EEG measurements). Our methods are underpinned with a principled statistical foundation, including in particular specialized ingredients from multiple hypothesis testing to inform our signal reconstruction procedure.
Further connections to related work \cite{RFDHB24,bobrowski2023universal}, are addressed later in the Discussion (Sec. \ref{s:discussion}).

\subsection{Time-frequency signal processing.} A more in-depth discussion calls forth a brief digression into time-frequency signal processing, a suite of techniques that distinguishes itself by studying square-integrable functions of time (known as \emph{finite-energy signals}), tracking frequency as it varies over time in the signal. 

One can think of someone whistling a melody they are reading on a musical score: time-frequency analysis aims at reconstructing the musical score from an audio recording of the whistler. 
This reconstructed musical score, known as a time-frequency representation of the signal, can then serve in many downstream tasks, such as detecting whether someone was indeed whistling, and if yes, identifying the song from which the melody is taken.
Time-frequency analysis is not limited to audio processing: a modern success story of the field is the detection of gravitational waves, which have a well-understood time-frequency signature in the interferometric signals collected by the LIGO and VIRGO detectors \cite{ligo2020guide}; other applications abound in a multitude of fields such as medical image processing \cite{zeng23}, time series \cite{xie2024}, analysis of financial data \cite{du2020imageprocessingtoolsfinancial}, biomedical signals \cite{wu2020current, xie2024} and biomechanics of postures \cite{xie2024}.

One of the dominant time-frequency representations is the short-time Fourier transform (STFT).
It is a complex-valued function of two variables, time and frequency, which collects Fourier transforms of the signal multiplied by a sliding window.  
Following the seminal work \cite{Fla}, a string of recent works \cite{meignen2016adaptive,BaFlCh18,BaHa19, koliander2019filtering, escudero2024efficient, ghosh2022signal, haimi2022zeros, miramont2023unsupervised, moukadem2024analytic} has focused on extracting information solely from the zeros of this transform, informally time-frequency pairs that are points of \textit{pure silence}; see also \cite{PaBa24Sub} for an introductory survey.
This approach is motivated in part by results that indicate that these zeros almost characterize the whole transform \cite{gardner2006sparse, Fla}, and in part by experiments hinting that detection using zeros comes with better statistical power than maxima-based methods at low signal-to-noise ratio (SNR) \cite{PaBa22,miramont2024benchmarking}. In short, the spectrogram zeros form a highly parsimonious, albeit non-linear, summary of the landscape of the STFT that is nonetheless informative enough to perform meaningful signal processing.

\subsection{Zeros, hyperuniformity and statistical mechanics.} Spectrogram zeros bring into play connections to the theory of point processes and strongly correlated particle systems from statistical mechanics, which we expound on below. To wit, when the signal is modeled by a random process,
the zero set of its STFT becomes a point process, i.e. a random configuration of points \cite{BaFlCh18, BaHa19, ghosh2022signal}. When the signal is pure white Gaussian noise, this point process turns out to correspond to well-known models of strongly correlated particle systems in statistical mechanics, namely the zeros of the so-called  planar (real) Gaussian analytic function (and its relatives) \cite{hough2009zeros}. In particular, such systems are known to exhibit a very high degree of homogeneity in their spatial distribution, a property known as \textit{superhomogeneity} or \textit{hyperuniformity} \cite{torquato2016hyperuniformity, abreu2017weyl, haimi2022zeros, salvalaglioAPS, torquato2023hyperuniformity, Milor_2025}. An illustrative sample can be found in Fig.~\ref{f:noise_spec_example_a}.

When a deterministic signal is added, the modulus of the STFT becomes large in a time-frequency region informally called the \emph{support} of the signal, effectively creating signal-shaped holes in the pattern of zeros; see e.g. Fig. \ref{f:noise_spec_example_c}.
Holes in a hyperuniform, grid-like pattern \cite{abreu2017weyl, torquato2023hyperuniformity} like the zeros of the planar Gaussian analytic function are intuitively easier to detect than similar holes in a completely random point pattern such as a Poisson point process.
This seminal observation of \cite{Fla} has motivated the use of techniques from spatial statistics to detect and extract signals by detecting and locating unusual holes in the pattern of zeros. Yet, such methods have so far relied on relatively rudimentary test statistics (such as pairwise distances \cite{BaFlCh18,PaBa22}) that often fall far short of capturing the intricate stochastic geometry of the zeros, often resulting in a collection of arbitrary algorithmic choices, leading to an ill-posed question of defining a significant hole from a collection of generic spatial statistics. This, therefore, calls for a principled, automated approach to detecting and reconstructing holes of arbitrary shape in the configuration of zeros of a spectrogram, which is a problem that we will focus on.
\\

\begin{figure}[h!]
    \centering
    \begin{subfigure}[b]{0.4\textwidth}
    \centering
        \includegraphics{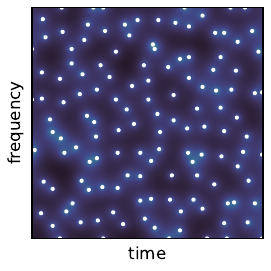}
        \caption{}
        \label{f:noise_spec_example_a}
    \end{subfigure}
    \begin{subfigure}[b]{0.4\textwidth}
        \centering
        \includegraphics{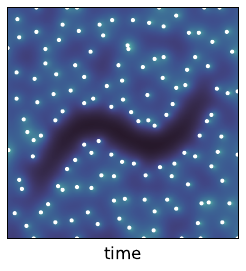} 
        \caption{}
        \label{f:noise_spec_example_c}
    \end{subfigure}
    \hfill    
    
    \caption{\textbf{(a)} Spectrogram of a realization of complex white Gaussian noise (white dots indicate the spectrogram zeros). \textbf{(b)} Same as \textbf{(a)} but with an additive mixture of signal and noise.}
    \label{f:noise_spec_example}
\end{figure}

\subsection{TDA meets signal processing.}
The inherent capability of TDA to access the finer topological properties of a point cloud, and the intrinsic requirement for this capability in signal analysis based on spectrogram zeros, make the intersection of these two domains a rich and natural overlap for symbiotic development -- a direction that we pursue in this work. To sum up, the crucial problem in zero-based signal processing is to algorithmically detect and reconstruct \textit{statistically significant} holes in an efficient manner; on the other hand, the algebraic topology of the two-dimensional spectrogram zero point cloud is fundamentally determined by its holes (in technical terms, its \textit{homology}), which smoothly  dovetails the algorithmic developments in persistent homology. Reciprocally, this overlap also raises fundamental and interesting statistical problems in TDA (such as the determination of the number and identities of statistically significant holes from persistence diagrams), some of which we take up in this paper. 

Furthermore, as discussed at the beginning, this brings the powerful tools and techniques of TDA to bear on a very general class of time-indexed signal processing problems where a pronounced geometric shape structure is not present in the \emph{a priori} description of the problem. We hope that this work will inaugurate a systematic development of the mutual interconnections between TDA and signal processing, motivating the growth of new tools and techniques in both fields.

\section{Spectrograms and their zeros}
\label{s:spectrograms}

The short-time Fourier transform (STFT) is a fundamental concept in time-frequency analysis \cite{flandrin1998time, grochenig2001foundations}. For a signal $f\in L^{2}(\mathbb{R})$ and a window $g\in L^{2}(\mathbb{R})$, the STFT of $f$ with window $g$ is defined as
\begin{equation}
\label{e:stft}
    V_{g}(f)(u,v) = \int_{-\infty}^{+\infty} f(t) \overline{g(t-u)}\; \e^{- 2\i\pi tv} dt,
\end{equation}
where the variables $u, v\in\mathbb{R}$ are usually referred as \emph{time} and \emph{frequency}.
Taking $\Vert g\Vert =1$, $V_g$ is an isometry, and the \emph{spectrogram} $(u,v) \mapsto |V_{g}(f)(u,v)|^{2}$ is considered a measure of how the \emph{energy} $\Vert f\Vert^2 = \Vert V_g(f)\Vert^2$ is spread across the time-frequency plane $\mathbb{R}^2$ \cite{flandrin2018explorations}.

In the fundamentally important case when the window is taken to be the Gaussian $g:t\mapsto 2^{1/4} \exp(-\pi t^2)$, $V_g$ maps tempered distributions $\mathcal{S}'(\mathbb{R})$ to analytic functions of $z=u+\i v$ \cite{AsBr09,BaFlCh18}. 
In particular, since analytic functions have isolated zeros, one can talk about the random configuration of isolated points (in probabilistic terms, the \emph{point process}) formed by the zeros of \(z =u+\i v \mapsto V_g(\xi)(u, v), \) where $\xi$ is a Gaussian white noise.

More specifically, one can define real Gaussian white noise as follows.
The Bochner-Minlos theorem \cite{holden1996stochastic} ensures the existence of a unique probability measure $\mu_1$ on $(\mathcal{S}^\prime(\mathbb{R}),\mathcal{B}(\mathcal{S}^\prime(\mathbb{R}))$ such that
\[ \mathbb{E}_{\mu_1}[\e^{\i\langle \cdot,\phi \rangle}] \coloneqq \int_{\mathcal{S}^\prime(\mathbb{R})} \e^{\i\langle \xi,\phi \rangle} d\mu_1(\xi) = \e^{-\frac{1}{2}\Vert\phi\Vert^2}, \phi\in\mathcal{S}(\mathbb{R}), \]
where $\Vert\phi\Vert^2 = \Vert\phi\Vert^2_{L^2(\mathbb{R})}$, $\langle \xi,\phi \rangle = \xi(\phi)$ is the action of $\xi\in\mathcal{S}^{\prime}(\mathbb{R})$ on $\phi\in\mathcal{S}(\mathbb{R})$.
$\mu_1$ is called real Gaussian white noise; in particular, for $\xi \sim \mu_1$ and orthonormal functions $\phi_1,\dots,\phi_n \in \mathcal{S}(\mathbb{R})$, the vector $(\langle \xi,\phi_1 \rangle,\dots,\langle \xi,\phi_n \rangle)$ is a real standard Gaussian random vector.
Now, for two independent real white Gaussian noises $\xi_1$ and $\xi_2$, the (complex) Gaussian white noise is defined as the law of $(\xi_1 + \mathrm{i} \xi_2)/\sqrt{2}$, which we define to apply to a pair of Schwarz functions $(\phi_1,\phi_2)$ by
$$
    \langle\xi_1,\phi_1\rangle + \mathrm{i}\langle{\xi_2, \phi_2}\rangle.
$$
Because the modulus of the STFT is proportional to the modulus of the so-called Bargmann transform, which maps the basis $(h_k)$ of Hermite functions to monomials \cite{grochenig2001foundations,gardner2006sparse,Fla}, one can actually prove that $V_g(\xi)$ has the same zeros as
the random series 
\[
    \sum\limits^{\infty}_{k=0}\left\langle\xi,h_k \right\rangle \frac{\pi^{k/2}z^{k}}{\sqrt{k!}},
\] 
which is entire $\mu_1$-almost surely; see \cite{ BaFlCh18,BaHa19}.
Since $(\left\langle\xi,h_k \right)$ are independent standard Gaussians, we recognize the planar Gaussian analytic function \cite{hough2009zeros}, the zeros of which are hyperuniform.

Figure~\ref{f:noise_spec_example_a} and Fig.~\ref{f:noise_spec_example_c} show the log spectrogram of resp. a realization of Gaussian white noise and a deterministic signal plus Gaussian white noise, with white dots indicating the zeros of the spectrogram.
As first noted by \cite{Fla}, the zeros are absent in the signal's time-frequency support, so that identifying the regions with low density of zeros is tantamount to identifying areas of the time-frequency plane where the signal is strong.


\section{Tools from topological data analysis} 
\label{s:tda}

Persistent homology is arguably the main workhorse of topological data analysis, providing a principled understanding of the multiscale geometry of a data point cloud, and significantly, enabling algebraic manipulations that are amenable to computer implementation. In the limited scope of the main text, we provide a succinct and conceptual overview focusing mostly on the 2D setting relevant for us, directing the interested reader to Section \ref{APP:s:tda} of the Appendix for a more detailed and in-depth discussion. For even more extensive coverage, we refer to any of the excellent texts \cite{chacholski2018building, ChBe21, EdHa22, dey2022computational}, to provide a partial list of references. 

\subsection{Persistence and Persistence Diagrams} 
\label{s:pers_homology-1}

\begin{figure}[h!]
    \begin{subfigure}{0.32\textwidth}
        \hspace{-1cm}
        \includegraphics{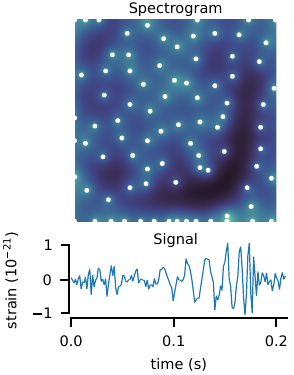}
        \caption{} \label{fig:pipeline_a}
    \end{subfigure}
    \begin{subfigure}{0.4\textwidth}
    \hspace{-1cm}
        \includegraphics{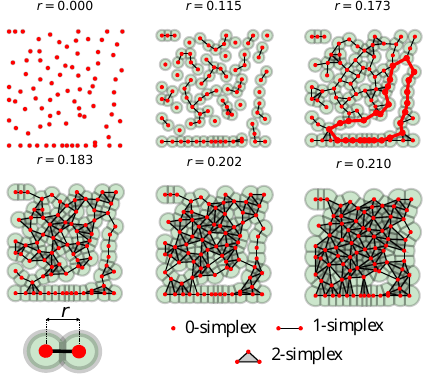}
        \caption{} \label{fig:pipeline_b}
    \end{subfigure}
    \begin{subfigure}{0.26\textwidth}
        \includegraphics{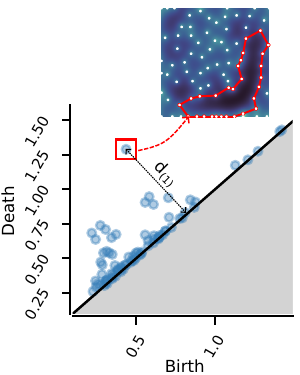}
        \caption{} \label{fig:pipeline_c}
    \end{subfigure}
    \caption{\textbf{(a)} The spectrogram of the gravitational wave event GW150914 (top), recorded by the Laser Interferometer Gravitational-wave Observatory (LIGO) in Hanford, WA, USA. and the corresponding signal (bottom). \textbf{(b)} Filtration of $\alpha$-complexes showing the birth of the most persistent cycle in red line. \textbf{(c)} Persistence diagram. The most persistent component is marked in red and the corresponding cycle is superimposed to the spectrogram for reference. $d_{(1)}$ is the longest distance to the identity diagonal.}
    \label{fig:pipeline}
\end{figure}

A very canonical approach to TDA proceeds via developing a \textit{geometric graph} (more precisely, a \textit{simplicial complex}) based on the given data point cloud, by connecting data points that are closer to each other than a threshold $r\geq 0$. For contextualization, we refer the reader to the illustrations in Figs. \ref{fig:pipeline_a} and \ref{fig:pipeline_b} (next page).
The threshold $r$ will be varied continuously, giving rise to a multiscale sequence of nested geometric graphs (a \emph{filtration}), which at any particular threshold value gives a skeletal representation of the dataset at that scale. 

Starting with the isolated data points when the threshold $r$ is equal to $0$, the graph grows and connectivity increases with increasing $r$, eventually culminating in the fully connected graph (c.f. Fig. \ref{fig:pipeline_b}). 
At some point, cycles start forming due to edges being connected; this stage may be thought of as the \textit{birth} of a cycle (to be thought of as a \textit{hole}). At this point, the $r/2$ disks centered at the endpoints of each edge of the cycle overlap pairwise, but all the disks pertaining to the cycle do not overlap jointly. At a suitably higher threshold $r$, all these disks do indeed overlap, in which case we declare the \textit{hole} to be \textit{filled} and the cycle to have \text{died}. It is reasonable to surmise that there are two kinds of cycles aka holes -- those that get \textit{filled} i.e. die soon after they are created, and those that \textit{persist} for a long time between their birth and death. Persistent homology theory identifies the latter variety of cycles or holes to be fundamental characteristics of the topological structure of the dataset, relegating the former type to be an artifact of random noise. The so-called \textit{Persistence Diagram} (\textit{abbrv.} PD) plots the times of birth and death of the different cycles against each other in a planar diagram; as such, the points in the PD that are far from the diagonal are understood to carry fundamental structural information on the dataset.
Figure \ref{fig:pipeline_c} displays the PD for the gravitational wave signal from Fig. \ref{fig:pipeline_a}, where is possible to appreciate that the most persistent hole, i.e. the point further from the identity in the PD, corresponds with the hole created by the signal in the pattern of zeros. 

\subsection{Persistent Homology and Learning Holes} 
\label{s:pers_homology-2}

The homological ingredient in Persistent Homology (\textit{abbrv.} PH) brings into play a crucial linear algebraic toolbox by considering the abstract vector spaces spanned by the cycles, edges and vertices in this (sequence of) geometric graphs, and by looking at the \textit{boundary maps} that connect each of these objects to their extremities. Under the standard \textit{homotopy equivalence} of cycles in algebraic topology, one can obtain natural (finitely generated) Abelian group structures on these vector spaces with the boundary maps inducing homomorphisms connecting them, with the groups corresponding to different thresholds being further connected via canonical maps arising out of natural embeddings of the simplicial complexes.
An in-depth description of this can be found in Sec. \ref{APP:s:pers_homology} of the Appendix.
 
The upshot of this is that efficient linear algebraic methods can be invoked to study these homology groups, wherein a non-trivial homology group (of \textit{first order}) corresponds to the presence of holes that equal in number to the rank of this group (counted modulo homotopy equivalence). In particular, there are effective tools to compute the \textit{generators} of these homology groups, which themselves correspond to (homotopy equivalence classes of) holes. This motivates our broad approach to the identification of non-trivial holes in a point cloud: identify points on the PD that are far from the diagonal (in a statistically significant sense).
Such an \textit{outlier} in the PD would pertain to a cycle in the relevant homology group, which would lead us to a hole in the underlying point cloud.

\subsection{Optimal Cycle Representations} 
\label{s:pers_homology-3}
A key issue in the above approach to learning holes is the inherent ambiguity in the homological cycles due to them being merely homotopy equivalence classes of cycles in the original point cloud (as opposed to being the cycles themselves). This leads to the question of identifying suitable representative cycles in these homological equivalence classes, culminating in the notion of \textit{Minimum} and \textit{Stable}  Volumes \cite{Oba18, Oba23}, that we will use in our algorithmic considerations.

\subsubsection{Minimum Volumes} It is intuitively appealing to search for a representative cycle for the equivalence class that encloses the smallest possible area of $\mathbb{R}^2$ in the spectrogram, that portion would likely correspond to the signal support. One way to formalize this minimizing class representer is the notion of \textit{minimum-volume} cycle, introduced in \cite{Sch15} and generalized in \cite{Oba18}. While the original definition takes the form of the solution to a constrained optimization problem, the solution of which has to be numerically approximated, our particular setup allows us to easily obtain minimum-volume cycles through the so-called \emph{persistence tree} (PT) of the point cloud, a notion introduced by \cite{Sch15}. 
Notice that in our application, we are more interested in the area enclosed by the minimum-volume cycles, henceforth termed minimum volume (\textit{abbrv.} MV).
For an algorithm to generate the PT, we refer the reader to the Algorithm \ref{a:persistence_tree} in the Appendix and \cite{Oba18}. A detailed discussion on MVs is taken up in Sec. \ref{APP:s:minvol_cycles} of the Appendix.

\subsubsection{Stable Volumes} It is a well-known issue in computational topology that a small change in the input signal, such as random noise, can produce large variations of the cycles --and their volumes--, even for the cycles associated to long-living components \cite{cohen2005stability,Ben20,Oba23}.  
Approaches to \emph{stabilize} the components of interest have been proposed \cite{Ben20,Oba23}. Intuitively, one way to robustify the construction of minimum volumes is to take the intersection of minimum volumes corresponding to slightly altered inputs, where the alteration is small enough that we can identify a volume from one signal to the other. Formally, we focus here on the so-called \emph{stable volumes} (\emph{abbrv.} SV) recently proposed by Obayashi \cite{Oba23}, which can be defined directly from the persistence tree described in Sec. \ref{APP:s:minvol_cycles} of the Appendix. Moreover, a detailed discussion on stable volumes and its computation is undertaken in Sec. \ref{APP:s:stable_vols} of the Appendix as well.

\section{Ingredients from Statistical Inference}
We give in this short section a few elements from statistical inference that we later use in signal detection and reconstruction. 
We heavily lean on multiple hypothesis testing to make informed decisions in our signal processing algorithms.
In simple hypothesis testing, one has one null hypothesis $H_{0}$ and an associated $p$-value for a given observation.
We say we reject $H_0$ when the $p$-value is lower than a prescribed significance level $\alpha$.
Intuitively, $\alpha$ is the probability of incorrectly rejecting $H_0$ (i.e. type I error) that, as a researcher, one is willing to accept.

Multiple hypothesis testing is a more general scenario where several tests are performed for a number of null hypotheses.
Say we perform $K$ tests, this is akin to having $K$ null hypotheses $H_{1,0},\dots,H_{K,0}$ and determining if any of the $K$ $p$-values is below some per-test significance level $\alpha_k$.

A relevant concept in this case is the \emph{family-wise} error rate (FWER), which is the worst-case probability of rejecting a true hypothesis among $H_{1,0}, \dots, H_{K,0}$, i.e. the probability of making \emph{any} type I error while conducting all $K$ tests \cite{wasserman2013all}.
Controlling the FWER means ensuring that the chance of falsely rejecting \emph{any} null hypothesis is less than~$\alpha$ \cite{dudoit2008multiple, fromont2016family}.
If $\alpha_k \equiv \alpha$, then from basic probability we know that the FWER is higher than $\alpha$ itself. Thus, some sort of \emph{correction} is needed.
A fairly used approach is the Bonferroni correction, where $\alpha_k \equiv \frac{\alpha}{K}$, which we shall use in the following sections.

For a more elaborate discussion of the concepts involved in multiple testing, we refer the reader to Sec. \ref{APP:s:hyp_test} of the Appendix.

\section{Signal detection via Persistence Diagrams} \label{s:test}

We formulate several detection tests based on different choices of test statistics. 
We first describe the simple and multiple testing procedures for a generic test statistic, to be instantiated later with test statistics based on TDA.

\subsection{Simple and multiple hypothesis testing}
Signal detection can be stated as a test aiming at rejecting $H_{0}: h = \xi$, the hypothesis that the measured signal is pure noise \cite{flandrin1998time,BaFlCh18,ghosh2022signal}.

\subsubsection{Simple hypothesis test}
Let $X$ be a generic test statistic that takes higher values for signals than for pure noise. 
We take $B$ independent Monte Carlo simulations of white noise and compute the test statistics $X_1,\dots,X_B$. 
A $p$-value for testing $H_0$ is then computed as $p = \frac{1}{B}\sum_{j=1}^B\mathbb{I}(X_{(j)}>X_{\text{obs}}),$ where $X_{(j)}$ are the top order statistics of $X$ and $X_\text{obs}$ is the test statistic obtained from the sample to be tested.

\subsubsection{Multiple hypothesis testing - simultaneous tests}
\label{s:simultaneous tests}
We devise an alternative test for detection by considering the $K$ most persistent components in the PD rather than just the most significant one as before.
Start with a suitable choice of $K$.
Again, from $B$-many simulated persistence diagrams, with statistics $X_{(k), j}, 1 \le j \le B$, we compute a series of $p$-values $p_1, p_2, \dots, p_{K}$ as 
\begin{equation} \label{e:p_vals}
    p_k := \frac{1}{B}\sum_{j=1}^B \mathbb{I}(X_{(k),j} > X_{(k), \mathrm{obs}}).
\end{equation}
We conduct these tests at varying levels of significance $\alpha_k$. 
A signal is detected if and only if $p_k<\alpha_k$ for some $1\le k\le K$. We choose $\alpha_k=\frac{\alpha}{K}$ according to the Bonferroni correction.
More details of this testing strategy can be seen in Sec. \ref{APP:s:hyp_test:mult:simul} of the Appendix.

The following subsections define the test statistics we derive from TDA: distance to the diagonal in the PD and energy from the most persistent components.

\subsection{Distance-based test statistic} \label{s:test_tda}
Formally, for a cycle $c$, let $d(c) = (y(c) - x(c))/\sqrt{2}$ denote the distance of the birth-death pair $(x(c), y(c))$ to the diagonal line in the PD.
We compute the order statistics $d_{(1)} \ge d_{(2)} \ge \cdots$ of these distances, where $d_{(1)}$ is the longest distance to the diagonal in the PD (as shown in Fig. \ref{fig:pipeline_c}), $d_{(2)}$ is the second longest and so on.
Possible test statistics in this case are $d_{(1)}$ (for simple hypothesis testing) or the $K$ most significant distances (for multiple hypothesis testing).

\subsection{Energy-based test statistic}
\label{s:energy test}
Aside from $d_{(k)}$, we have derived an alternative test statistic based on the energy $\mathcal{E}_{(k)}$ contained in the $K$ most persistent components according to the distances to the diagonal line in the PD.

In practice, spectrograms are approximated by Riemann sums at a finite grid of times and frequencies, yielding a matrix $\mathbf{S}=(s_{mn})\in \mathbb{R}_{+}^{M\times N}$, where $m\in\{0,1,\dots,M-1\}$ is the frequency bin and $n\in\{0,1,\dots,N-1\}$ is the time index.
Thus \(\mathbf{S} = s_{mn} \approx |V_g(f)(\frac{n}{N}T_{\text{s}},\frac{m}{2MT_{\text{s}}})|^2\) for a sufficiently small sampling period $T_{\text{s}}$.

We compute $\mathcal{E}^{\text{MV}}_{(k)}$ (resp. $\mathcal{E}^{\text{SV}}_{(k)}$), the energy contained in the minimum (resp. stable) volume of the $k$-th most significant component, as
\[\mathcal{E}^{\text{MV}}_{(k)} = \frac{\sum_{m,n} (\mathbf{1}^{\text{MV}(k)}_{mn} s_{mn})}{\sum_{m,n} s_{mn}}, \; \mathcal{E}^{\text{SV}}_{(k)} = \frac{\sum_{m,n} (\mathbf{1}^{\text{SV}(k)}_{mn} s_{mn})}{\sum_{m,n} s_{mn}},\]
where $(\mathbf{1}^{\text{MV}(k)}_{mn})$ and $(\mathbf{1}^{\text{SV}(k)}_{mn})$ are the binary masks indicating the time-frequency region in the spectrogram occupied by the minimum volume and the stable volume of the $k$-th most significant component, respectively.

\subsection{Energy vs distance based tests} Depending on the nature of the dataset, hypothesis testing based on energies might outperform hypothesis testing based on distances to the diagonal. 
For instance, signals containing a noticeable amount of energy but in a short duration and narrow frequency range occupy a small region in the time-frequency plane that is less detectable using distance statistics. 
The comparison between the two choices of test statistics is further elaborated and demonstrated in the experiments sections (Secs. \ref{s:applications-detection} and \ref{s:applications-reconstruction} ).

\subsection{Accumulated persistence function} \label{s:test_apf}

In \cite{biscio2019accumulated}, Biscio et al. proposed a Monte Carlo envelope test based on a function of the PD called the \emph{accumulated persistence function} (APF).
APF condenses the information from the persistence diagram obtained from the observation, or a simulation, into a single real-valued function.
We use a Monte Carlo \emph{global} test \cite{baddeley2014tests, myllymaki2017global} based on contrasting the statistics of the APF of simulations and data in order to detect a signal.
Details of this test can be seen in the Sec. \ref{APP:APF} of the Appendix.

\section{Signal reconstruction via Persistence Diagrams} \label{s:signal_estimation}
\begin{figure}
\centering
    \includegraphics{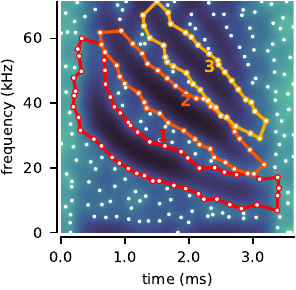}
    \includegraphics{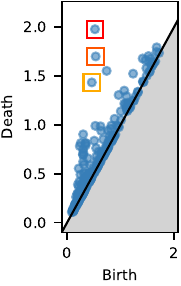}
    \caption{Example of component identification on a bat echolocation signal. The number superimposed on each component indicates the relevance of the component in the persistence diagram shown on the right.}
    \label{f:batsig}
\end{figure}

Finding persistent holes in the zeros of the spectrogram allows the identification of regions where the signal is stronger than noise, the so-called \emph{signal domain} $D$ in the TF plane \cite{Fla,BaFlCh18}.
In this section we explain how to reconstruct a signal from its noisy samples based on an estimation of $D$.


\subsection{Signal reconstruction}
In order to properly estimate $f$, we first need to determine the number of components, i.e. persistent holes in the pattern of spectrogram zeros, which we use to approximate $D$.
As an example, Figure \ref{f:batsig} shows the minimum-volume cycles corresponding the three most persistent components of a bat echolocation signal, and their corresponding birth-death pairs in the PD.
Ideally, $D$ corresponds to the union of the associated minimum (or stable) volumes.
Once $D$ is estimated, the signal itself can be estimated as \cite{Fla}:  
\begin{equation} \label{e:inv_stft}
    \hat{f}(t) = \frac{1}{g(0)}\int_{-\infty}^{+\infty} V_{g}(h)(t,\nu) \mathbb{I}_{D}(t,\nu) e^{2\i\pi t\nu} d\nu.
\end{equation}

This leads to the question of how to estimate the number of holes, which we address in the following section.

\subsection{Estimation of the number of holes} \label{s:number_holes}
Counting the number of holes is tantamount to the problem of detecting several signals.
Hence, we use once again multiple hypothesis testing but now with the goal of estimating an appropriate number of relevant components, rather than just claim the existence of a signal as we did in Sec. \ref{s:test}. Our approach to the estimation of the number of holes is underpinned by mathematical analysis carried out in Sec. \ref{APP:s:hyp_test:mult:seq:alpha_levels} of the Appendix.

\subsubsection{Multiple hypothesis testing - sequential tests}
As before, from the $B$-many noise realizations we obtain independent test statistics $X_{(k), 1}, \ldots, X_{(k), B}$ and compute $K$ $p$-values $p_{k}$ using \eqref{e:p_vals}.
Now, in contrast to what we do for simultaneous tests, we conduct these tests \emph{sequentially}, i.e. one-by-one, at varying levels of significance $\alpha_{k}$. 
If $p_1$ is significant, there is evidence for the presence of signal (as in the simple hypothesis case). 
For $k \ge 1$, if $p_{k}<\alpha_{k}$ we go on to conduct the test for $k+1$. 
We stop at $k_*$, the smallest $k$ such that $p_{k} \geq\alpha_{k}$, declaring that our estimate of $n_{\mathrm{holes}}$, the total number of holes or persistent components to extract, is $k_* - 1$. 
Thus $\hat{n}_{\mathrm{holes}} = \min \{ k \ge 0 : p_{k+1} \ge \alpha_{k + 1}\}$.
Notice that the $p$-values are computed for up to $K$ components. 
The tests can be carried out by taking the test statistics to be the distances $d_{(k)}$ to the diagonal line in the PD, or the energies $\mathcal{E}^{MV}_{(k)}$ (or $\mathcal{E}^{SV}_{(k)}$) with $p$-values computed in a similar way as before. 

\subsubsection{Choice of levels of significance}
We consider three different choices of $\alpha_k$ for sequential testing: (i) Polynomial decay: $\alpha_k = \frac{\alpha}{k^{m}},\quad m \ge 0$; (ii) Geometric decay: $\alpha_k = \alpha \beta^{k}, \quad \beta \in (0, 1)$ and (iii) Bonferroni correction.
Both polynomial and geometric decay can be shown to control the FWER under the assumption that the $p$-values are independent.
Further details of such corrections can be seen in Sec. \ref{APP:s:hyp_test:mult:seq} of the Appendix.

\section{Applications: Detection of synthetic and real-world signals}
\label{s:applications-detection}

This section describes the application of the proposed approaches for signal detection based on TDA and the zeros of the spectrogram.
First, the results of the tests for signal detection considering a synthetic signal are reported.
Then we illustrate the performance of the detection tests on gravitational wave signals.
Alpha complexes and persistence diagrams were computed using GUDHI's Python front-end \cite{gudhi:urm}. 
The code to reproduce the experiments in this paper is provided\footnote{\url{https://github.com/rbardenet/persistence}}.

Finally, we define the signal-to-noise ratio (SNR) of a signal-plus-noise mixture $h=f+\xi$ as
\begin{equation}
    \text{SNR}(h) = 10 \log_{10} \left(\Vert f\Vert^2/\mathbb{E}\left[\Vert \xi \Vert ^{2}\right]\right)\quad \text{(dB)}.
\end{equation}

\subsection{Synthetic chirp detection} \label{s:results_det}
\subsubsection{Test signal specification}
Following similar experiments from \cite{BaFlCh18, PaBa22, PaBa24Sub}, we first demonstrate the performance of the detection tests described in Sec. \ref{s:test}.
To that end, we consider an alternative hypothesis $H_1$ to $H_0$, where a signal is hidden in additive noise. 
For the signal, we consider a discrete linear chirp length $N=1024$, that is a signal whose instantaneous frequency increases linearly with time.
The chirp has duration $N/2$ and is centered in the observation window.
We can then estimate the \emph{power} of each test, i.e. the probability under the model consisting of the chirp plus white noise of rejecting $H_0$.
Figure \ref{f:results_det} shows the power of different tests for several SNRs.
The spectrogram of the probed signal is also shown in the figure.

\begin{figure}[h!]
    \centering
        \includegraphics[width=0.5\textwidth]{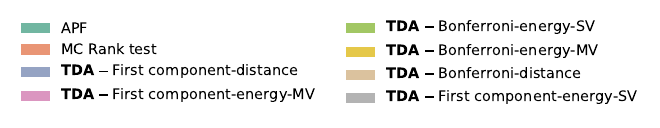} \\    
         \includegraphics[scale=1.0]{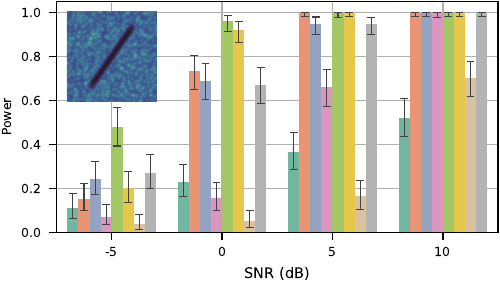} 
        \caption{Statistical power of the detection tests described in Sec. \ref{s:test} as a function of the SNR. Each color and marker indicate a different test, the name of which is indicated on top of the figure. Error bars indicate the $95\%$ Clopper-Pearson confidence intervals, Bonferroni-corrected across the 32 comparisons (4 SNR levels $\times$ 8 approaches). The spectrogram of the probed signal with $\text{SNR}=10$ dB is superimposed.}
        \label{f:results_det}
    \end{figure}

\subsubsection{Detection tests}
We compared a total of eight tests: six variations of TDA-based tests (Sec. \ref{s:test}), a test based on the APF (Sec. \ref{s:test_apf}), and a Monte Carlo test based on ranked envelopes (taken from \cite{miramont2024benchmarking}).
We consider the latter two as a baseline for signal detection tests based on spectrogram zeros.

Two TDA-based tests rely on the distance to the diagonal in the PD (Sec. \ref{s:test_tda}). 
One uses the distance to the diagonal of the furthest birth-death pair in the PD, i.e. the first or most relevant component (denoted ``First component-dist'' in Fig. \ref{f:results_det}).
The other searches for a statistically relevant component among the $K=30$ components with the largest distance to the diagonal, using a Bonferroni correction to control the type I error (``Bonferroni-dist'' in Fig. \ref{f:results_det}).

The remaining four tests based on TDA use the energy associated with the components from the PD as test statistic (Sec. \ref{s:test_tda}).

The energy was computed from the MVs or the SVs (defined in Sec. \ref{s:pers_homology-3}) corresponding to the first $K$ persistent components.

\subsubsection{Detection results}
Figure \ref{f:results_det} shows that the proposed tests based on the energy of the first $K=30$ significant components have the highest detection power from all the evaluated tests (``Bonferroni-energy-SV'' and ``Bonferroni-energy-MV'').
From these, the tests based on stable volumes (``Bonferroni-energy-SV'') are more powerful than the one based on minimum volumes, specially for low SNRs.
These tests have a better performance than the baseline (``MC Rank test'') for all the considered signal lengths.
The two other tests based on the energy (``First component-energy-MV'' and ``First component-energy-SV'') have a very similar performance.

In comparison, the tests based solely on the distance to the diagonal (``First component-dist'' and ``Bonferroni-dist'') and the test based on the APF show lower power than the baseline for the detection of a single component, even for high SNRs.

\subsection{Gravitational waves} \label{s:gw_detection}

\subsubsection{The Gravitational Waves problem}
Gravitational waves (GWs) are spacetime perturbations produced by astrophysical events such as merging black holes. 
Their detection enables tests of general relativity and studies of extreme cosmic phenomena. However, detecting GWs is a formidable challenge because of their exceptionally low SNR.

For illustration purposes, we have selected GW data from the \textit{Kaggle Gravitational Waves Detection competition} \cite{kaggle-g2net}. 
The dataset consists of realistic simulations of measurements from interferometers.
Each data sample contains a signal spanning 2 seconds and sampled at $2048$ samples-per-second.

\subsubsection{Stratifying GW signals via adaptive thresholding}
Since SNR values are not available, we compute a hard threshold \cite{donoho1994ideal,mallat2008wavelet, pham2018novel} on the spectrogram of each signal in order to segment the dataset.

We define the normalized spectrogram maximum (NSM), $\hat{\zeta} = \hat{\gamma} / \hat{\sigma}$ , where $\hat{\gamma}=\max_{u,v}|V_g(f)(u,v)|$ is the maximum modulus of the STFT across the observation window on the time-frequency plane, and $\hat{\sigma}$ is the standard deviation of the noise, estimated from the signal using a median absolute deviation estimator \cite{donoho1994ideal, mallat2008wavelet}.

\begin{figure}[h!]
    \centering
    \includegraphics{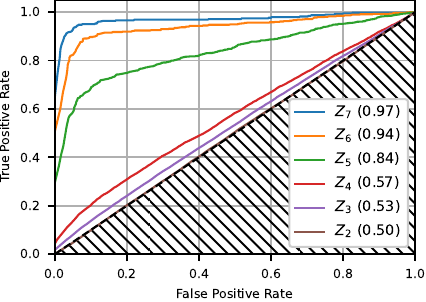}
    \caption{ROC curves for different thresholds $\zeta\in\{2,3,4,5,6,7\}$ using energy as test statistics. The values shown in the parentheses are the area under the ROC curves for each threshold.}
    \label{f:ROC_b}
\end{figure}

The NSM serves as a practical proxy for SNR; signals with higher NSM values are expected to have higher SNR values.
We then consider the class of GW signals, which is the union of the subsets $Z_i$ for $i \in \{2, 3, 4, 5, 6, 7\}$, where each $Z_i$ contains signals with $\hat{\zeta} \ge i$. Note that $Z_i\subset Z_j$ if $i\ge j$.

We collected 50,000 signal samples for $Z_2$, from which $|Z_3|=49,995$, $|Z_4|=14,542$, $|Z_5|=1,528$, $|Z_6|=670$, and $|Z_7|=255$. 
Even with a threshold of $\hat{\zeta}= 5$, we retained a large enough number of signals for robust analysis.
Examples of spectrograms of the signals in the different subsets can be found in Sec. \ref{APP:s:gw_detection} of the Appendix.
Most signals in $Z_7$ are visually identifiable in their spectrograms, while $Z_5$ and even $Z_6$ already contain numerous more subtle signals.

We observe that our NSM thresholding is only used to create subsets of the data for comparison purposes; it plays no role in the detection procedure whatsoever. We also note that our approach is a general signal analysis technique that is being applied to the challenging problem of GW data as a particular application, as opposed to most state-of-the-art tools that are highly specialized and tailored to the GW problem. We envisage that combining our methods with specialized GW data analysis tools have the potential for producing powerful signal processing outcomes for GW data.

\subsubsection{Gravitational waves detection}

Figure \ref{f:ROC_b} shows the Receiver Operating Characteristic (ROC) curves for the tests based on $\mathcal{E}^{\text{SV}}_{(\ell)}$, $\ell=1,\dots,K$, the energies contained in stable volumes associated with the $K$ most significant components (see Sec. \ref{s:energy test}).

For this application, we chose $K=5$ and applied simultaneous hypothesis tests with Bonferroni correction (see Sec. \ref{s:simultaneous tests}).

ROC curves were computed using class-balanced datasets for each $Z_i$. As one would expect, signals from datasets corresponding to higher thresholds are likelier to have higher SNRs, making their features easier to detect.

We achieved an area under the ROC curve of $0.97$ for the highly curated dataset $Z_7$ (see Fig. \ref{f:ROC_b}).
For subset $Z_5$, which includes a large number of more challenging examples, the area under the ROC curve is still high, to the tune of $0.84$.
Including the examples with the lowest SNRs reduces the performance, as our generic approach is not tailored for GW detection---unlike other competing methods that rely on matched filtering. Future work will focus on enhancing the detection power of our tests (see the discussion in Sec. \ref{s:discussion}).

\subsection{Detection summary}
In Sec. \ref{s:results_det} we compared detection tests based on two different choices of test statistics. 
The first uses the distance to the diagonal in the PD as a test statistic, whereas the second uses the energy of the most persistent components as a test statistic.
The latter turned out to be more effective for signal detection (see Fig. \ref{f:results_det}). 
When comparing the power of the distance-based tests (``dist'' in Fig. \ref{f:results_det}) with those based on energy (``energy'' in Fig. \ref{f:results_det}), it seems that large values of the spectrogram play a key role in improving detection tests based solely on the zeros of the spectrogram.


\begin{figure}[h!]
    \centering
    \hspace{-0.7cm}
    \begin{subfigure}[b]{0.33\textwidth}
        \includegraphics{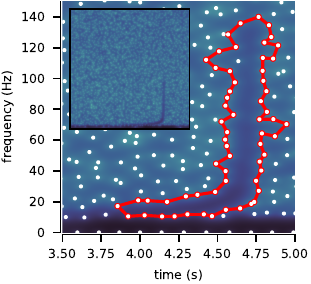}
    \caption{}\label{f:ligo_a}
    \end{subfigure} 
    \begin{subfigure}[b]{0.15\textwidth}
        \includegraphics{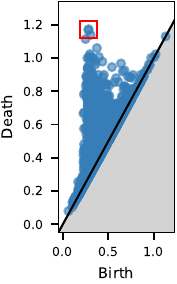}
    \caption{}\label{f:ligo_b}
    \end{subfigure}
    \\
    \begin{subfigure}[b]{0.5\textwidth}
        \includegraphics{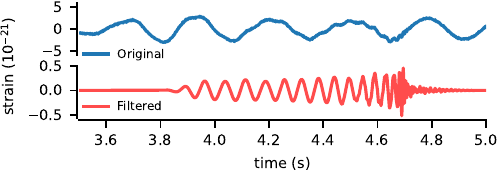}
    \caption{}\label{f:ligo_c}
    \end{subfigure}
    \caption{Example of gravitational wave detection and estimation. \textbf{(a)} Detail of the spectrogram of the signal, the boundary of the stable volume corresponding the detected component in red (the complete spectrogram is superimposed for reference). \textbf{(b)} Persistence diagram. The red square indicates the detected component. \textbf{(c)} Time domain comparison of the original and the detected, filtered signal. Notice the difference in the scale of the vertical axis.}
    \label{f:ligo}
\end{figure}

\section{Applications: Signal Reconstruction and Audio Improvement} \label{s:applications-reconstruction}

In this section we demonstrate how we can reconstruct signals by identifying the volumes in the TF plane corresponding to the most persistent components.
We first show an example of the reconstruction of a GW from the dataset used in Sec. \ref{s:gw_detection}.
Then we quantitatively characterize the performance of our approach with synthetic and audio signals.

\subsection{Gravitational Wave Reconstruction}
Figure \ref{f:ligo_a} shows the spectrogram of a gravitational wave taken from the dataset.
Only one component was detected, as expected, and Figure \ref{f:ligo_a} shows that the boundary of the corresponding stable volume (in red), which correctly encloses the signal domain of the gravitational wave with its characteristic shape of an \emph{exponential} chirp \cite{chassande1999time,ligo2020guide}.
Figure \ref{f:ligo_b} shows the associated PD, where the birth-death pair corresponding to the detected component is highlighted in red.
The original and estimated signals are shown in Fig. \ref{f:ligo_c}.

In contrast to \cite{Fla}, our approach based on persistent homology does not depend on a parameter that fixes the minimum scale of the holes (such as a minimum length of the edges of the triangles).
Instead, one needs to determine the number of components to be estimated from the signal.
This is automatically and adaptively done by the detection tests proposed in Sec. \ref{s:number_holes}, so that the method is fully based on the data.

\subsection{Synthetic Signal Reconstruction}
We show quantitative results for signal reconstruction using three synthetic signals with different time-frequency structures. Their spectrograms can be seen in Fig. \ref{f:synth_results}.

\begin{figure*}
    \includegraphics[width=\textwidth]{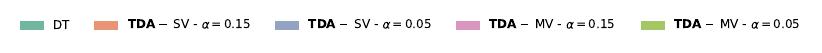}
    \begin{subfigure}[b]{0.34\textwidth}
    \hspace{-0.6cm} \includegraphics{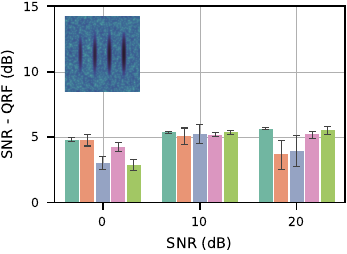}
    \caption{Impulses.} \label{f:synth_results_a}    
    \end{subfigure} 
    \begin{subfigure}[b]{0.32\textwidth}    \includegraphics{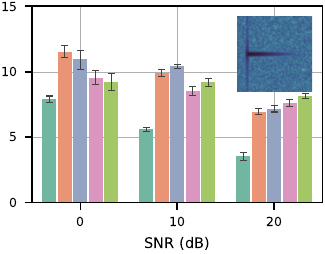}
    \caption{Sharp attack.} \label{f:synth_results_b}   
    \end{subfigure} 
    \begin{subfigure}[b]{0.32\textwidth}
    \hspace{0.4cm}
    \includegraphics{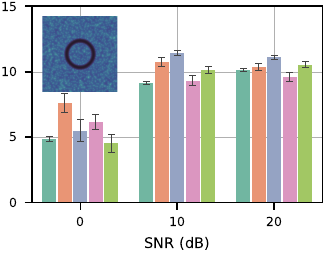}
    \caption{Hermite function.} \label{f:synth_results_c}    
    \end{subfigure}

    \caption{\textbf{(a)}-\textbf{(c)} Synthetic signal reconstruction performance measured by SNR gain, i.e. SNR-QRF (dB) for, using minimum and stable volumes. Bars indicate the average over $200$ realizations. Error bars indicate the $95\%$ confidence interval, Bonferroni-corrected across the 15 comparisons (3 SNR levels $\times$ 5 approaches). The proposed approach is indicated as \textbf{TDA}. DT stands for the Delaunay triangulation approach.} 
    \label{f:synth_results}
\end{figure*}

\subsubsection{Performance metric}
The reconstruction performance of synthetic signals was measured using the quality reconstruction factor \[\operatorname{QRF} = 10 \log_{10}\left(\Vert f \Vert ^2_{2} / \Vert f-\hat{f} \Vert^{2}_{2}\right)\; \text{(dB)},\]
where $f$ is the original noiseless signal and $\hat{f}$ is a denoised approximation \cite{meignen2016adaptive}.
QRF can be understood as an estimation of the SNR of the output signal after reconstruction, which should be higher than the SNR of the noisy input signal to indicate an improvement.
Therefore the difference $\text{QRF}-\text{SNR}$ (dB) is equivalent to the gain in terms of SNR after processing.
Notice that one can only compute the QRF if the noiseless reference is available, hence QRF is only reliably computable from synthetic signals.

\subsubsection{Results}
Figure \ref{f:synth_results} compares the gain in SNR i.e. $\text{QRF}-\text{SNR}$ (dB), using MVs and SVs for several input SNRs and three different signals. 
Notice that we also vary the significance of the detection tests used to count the number of components ($\alpha=0.05$ and $\alpha=0.15$).
Using SVs and a higher $\alpha$ results in a better reconstruction in the case of low SNRs (for instance, $0$ and $10$ dB).
In most cases the results obtained with TDA are better or comparable to those obtained with DT \cite{Fla}.

\subsubsection{Discussion of results}
The situation where one \emph{signal} component corresponds to more than one \emph{homology} component can occur in practice, particularly for low SNRs.
Detection tests that automatically determine the number of relevant homology components to extract alleviate this issue.
However, such tests are, in practice, limited by the SNR. 
If detection tests are not effective, the user can always determine the number of homology components to extract by hand. 
Even depending on \emph{a priori} detection of components based on detection tests, signal estimation performance with the approach based on TDA is comparable or better than previous approaches based on spectrogram zeros (compare with the DT method in Fig. \ref{f:synth_results}).
Increasing $\alpha$ results in a larger number of homology components detected, and it could be an acceptable strategy when a signal is known to be present beforehand, with reducing the noise the primary goal rather than detecting the signal (as done in \cite{Fla}).


\subsection{Audio improvement}

\begin{figure}

    \centering
    \includegraphics[width=0.5\textwidth]{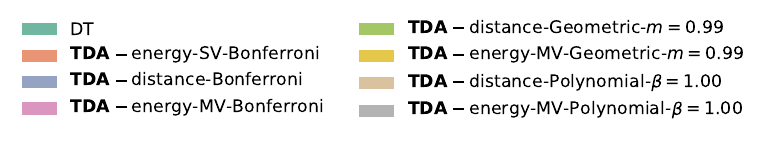} \\
    \includegraphics{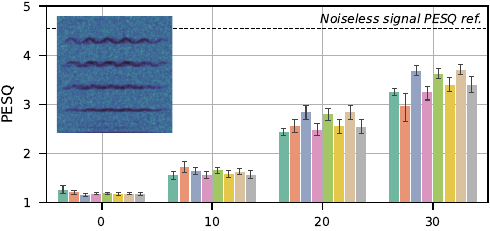}
    
    \caption{Output audio quality measured by PESQ \cite{recommendation2001perceptual} after applying variations of the proposed approach (indicated as \textbf{TDA}) and the Delaunay triangulation (DT) based method for a short cello recording. Error bars indicate the $95\%$ confidence intervals, Bonferroni-corrected across the 32 comparisons (4 SNR levels $\times$ 8 approaches). }\label{f:results_pesq}
\end{figure}
\begin{figure*}[h!]
    \resizebox{0.99\textwidth}{!}{%
    \begin{subfigure}[b]{0.43\textwidth}
    \includegraphics{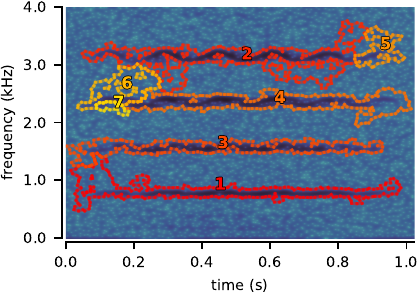} 
    \caption{Minimum-volume cycles.}\label{f:cello_detail_a}
    \end{subfigure}
    \begin{subfigure}[b]{0.43\textwidth}
    \includegraphics{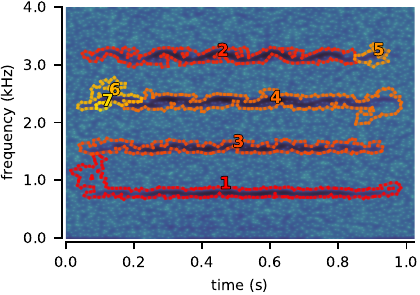}
    \caption{Stable volumes.}\label{f:cello_detail_b}
    \end{subfigure}
    \begin{subfigure}[b]{0.13\textwidth}
    \includegraphics{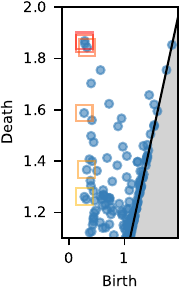}
    \caption{}\label{f:cello_detail_c}
    \end{subfigure}}
    
    \caption{Extracted volumes for the cello signal. \textbf{(a)} Minimum-volume cycles. \textbf{(b)} Boundary of stable volumes. The number indicates the relevance of the component in the persistent diagram. \textbf{(c)} Persistence diagram.}
    \label{f:cello_detail}
\end{figure*}

We now explore the reconstruction performance of an audio signal corresponding to a cello excerpt, sampled at $8$ kHz and contaminated with white Gaussian noise.
Figure \ref{f:results_pesq} shows the spectrogram of the signal with four harmonics.
In contrast to the synthetic case, there is no ground-truth, noiseless version of the real-world signal available.

\subsubsection{Performance metric}
We compare the audio quality of the estimated signal by means of the \emph{perceptual evaluation of speech quality} (PESQ) metric \cite{pesq_python, recommendation2001perceptual}.
Although originally conceived to assess speech quality over telecommunication networks, PESQ has been used as a general audio quality metric, even for non-speech signals \cite{torcoli2021objective}. 
A large value of PESQ is then interpreted as indicative of better audio quality.

\subsubsection{Results}
Figure \ref{f:results_pesq} shows the PESQ values for several SNRs (from $0$ dB to $30$ dB).
As a reference, the value of PESQ for the signal not contaminated with noise, is also shown.
We explored variations of our TDA-based approach method, modifying the procedure by which the number of components to extract is determined (see step 4 in Algorithm \ref{alg:sig_est} described in the Appendix).

The three $\alpha$ corrections from Sec. \ref{s:number_holes}, i.e. Bonferroni, geometric decay and polynomial decay, were pairwise combined with the two test statistics described in Sec. \ref{s:test}, i.e. $d_{(k)}$ and $\mathcal{E}^{SV}_{(k)}$. 
The significance level was fixed to $\alpha=0.15$ in this case.

For $\text{SNR}>10$ dB, Figure \ref{f:results_pesq} shows that when using ``TDA-distance'' the results tend to be better than when the energy is used as a test statistic (``TDA-energy'').
The choice of $\alpha$ correction seems to play a less relevant role.
In the low SNR regime, the use of $\mathcal{E}_{(k)}$ in combination with stable volumes increases the performance, similar to the results obtained for synthetic signals.

\subsubsection{Discussion of results}
The results obtained for the proposed approach are comparable to those obtained with the DT approach (``DT'' in Figure \ref{f:results_pesq}) in the case of low SNR, while for higher SNRs the results are better than with DT.
Figure \ref{f:cello_detail_a} shows the minimum-volume cycles identified for the cello signal.
Real-world signals might present an intricate structure in the spectrogram as well non-white noise (contrary to what it was assumed in the Monte Carlo simulations used to determine the number of components).
Such factors can produce instabilities in the MVs. 
However, stable volumes can alleviate this issue, as depicted in Fig. \ref{f:cello_detail_b}.

An important difference with preexisting approaches is that using TDA one can simultaneously identify individual components and their domain (see Fig. \ref{f:batsig} or Fig. \ref{f:cello_detail} in the main document) without a post-processing heuristic that groups triangles into components (as done in \cite{meignen2016adaptive}).

\begin{figure*}[h!]
    \hspace{-1cm}
    \begin{subfigure}{0.55\textwidth}
    \includegraphics{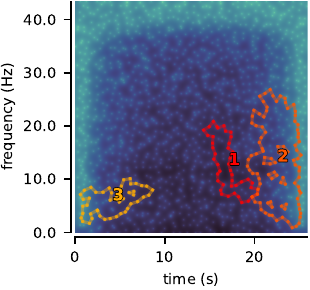}
    \includegraphics{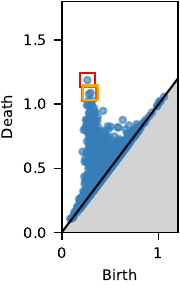}
    \caption{Single-channel, seizure-free EEG.}\label{f:seizures_a}
    \end{subfigure}
    \begin{subfigure}{0.55\textwidth}
    \includegraphics{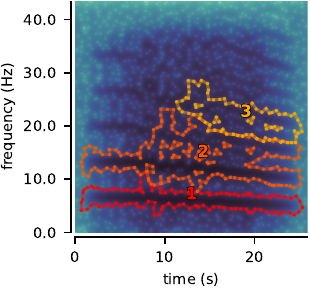}
    \includegraphics{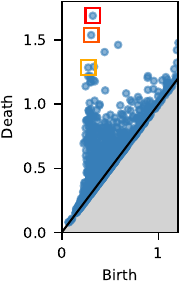}
    \caption{Single-channel seizure EEG.} \label{f:seizures_b}
    \end{subfigure}
    \caption{\textbf{(a)} Spectrograms of EEG with no seizure present. \textbf{(b)} EEG with a seizure. The minimum-volume cycles corresponding to the three more persistent components are shown in each case. The number superimposed on each component indicates the relevance of the component in the persistence diagram shown on the right of each subfigure.}
    \label{f:seizures}
\end{figure*}

\section{Further applications, discussion and future work} \label{s:discussion}
\subsection{Further applications}
A key feature of our approach is the general scope of its applicability. While its utilization on acoustic signals and gravitational wave signals has been discussed in earlier sections, we hint at the scope and versatility of our methods here with a possible usage in neuroscience, as a sample of other future applications to non-acoustic domains.
This application consists in the characterization of TF \emph{signatures} using TDA and the tools described before to extract \emph{features} that can be fed to downstream classification (or regression) models.

To demonstrate the potential of such a framework for the analysis of time-series, we describe an example based on non-acoustic signals from a intracranial electroencephalograms (EEGs) dataset \cite[see datasets C, D and E]{andrzejak2001indications}.
These EEGs can be classified as having a seizure present, i.e. pathological signals, or reflecting seizure-free brain activity.
Seizures have been associated with a more synchronized activity in EEGs, whereas typical brain activity could be succinctly described as noise-like \cite{andrzejak2001indications}.
Because of these differences, when a seizure is present, EEGs tend to be more organized and \emph{rhythmic}, which translates into a different signature in the TF plane consisting of more persistent holes in the pattern of spectrogram zeros than seizure-free EEGs.
Figure \ref{f:seizures} shows the difference between two signals obtained from the dataset and their corresponding PDs.
Figure \ref{f:seizures_a} shows the spectrogram corresponding to a seizure-free EEG, whereas Fig. \ref{f:seizures_b} corresponds to an EEG with a seizure. 
Comparing the persistence diagrams from Figs. \ref{f:seizures_a} and \ref{f:seizures_b}, the latter has more persistent holes, linked to a pattern of harmonic-like components in the spectrogram.  

\subsection{Discussion and future work}
To our knowledge, this work is arguably the first to systematically apply TDA as a broad-based methodological toolbox for spectrogram zero-based signal processing. The limited amount of existing literature that appear to be related to our work often involves specialised applications and particular kinds of signals. For instance, \cite{RFDHB24} investigates the problem of audio obfuscation, where topological features based on Betti numbers are invoked. In a different direction, \cite{guillemard2011groupoid} proposes a groupoid $C^*$ algebra-based method to combine dimensionality reduction and generic persistent homology for recovering topological features of the input signal, while \cite{boche2013signal} works on filtered level sets of the spectrogram for feature extraction. As alluded before, persistence diagrams have also been converted to a functional statistic called the \emph{accumulated persistence function} in spatial statistics \cite{biscio2019accumulated}.

While the APF can be used in conjunction with classical spatial statistical tests, our approach is more tailored to time-frequency detection and reconstruction, which we have numerically demonstrated (cf. the detection power of APF in \ref{f:results_det}). 

The signal detection methods described in this paper, as well as the detection tests introduced in \cite{BaFlCh18, PaBa22, ghosh2022signal}, fall under a different paradigm than traditional signal detection methods where a model of the signal is available, like matching filter techniques and variations thereof \cite{chassande1999time, whalen2013detection}.
In this sense, the signal detection test defined in Sec. \ref{s:test} constitutes a signal-agnostic alternative, relying on a model of the noise and Monte Carlo simulations rather than a parametric model of the signal.
Such tests might find future applications to the detection of so-called \emph{continuous} gravitational waves, for which matching possible signal templates is computationally ineffective \cite{ashok2021new}.
Furthermore, extending the proposal to the continuous wavelet transform and its zeros \cite{BaHa19} could further improve detection performance in certain applications.
Connections with recent work on hypothesis testing based on the PD \cite{bobrowski2023universal} are to be explored in future work, where applicability of certain \emph{universal} null models could be considered in order to increase the power of our tests.

\section{Conclusion} 
In this work, we propose a very general approach to the detection and reconstruction of signals from noisy measurements, leveraging a robust interplay between time-frequency analysis based on spectrogram zeros and state-of-the-art concepts and techniques from topological data analysis (TDA). Our approach is undergirded by a solid statistical foundation based on the theory and methods of multiple hypothesis testing, resulting in the possibility of providing scientifically principled predictions augmented with confidence guarantees. In terms of ideas, our method interfaces with a range of disparate scientific domains -- time-frequency analysis (which is a staple tool in the processing of a wide variety of time-dependent signals), hyperuniformity of the spectrogram zero patterns (which is of interest in statistical physics),  the theory of stable and minimum volumes and that of persistence trees (of recent interest in topological data analysis), and  conceptual ingredients like FWER and Bonferroni correction (which are of fundamental importance in statistical inference). 

In terms of applications, we are able to bring our methods to bear on a highly versatile class of time-varying signals, including acoustic data from musical instruments (of key interest in the field of audio and speech recognition), Gravitational Wave (GW) data from the celebrated LIGO experiment (a burgeoning area of interest in physics and astronomy), and futuristic applications towards EEG time series data (of interest in medical science). Our topologically-based methods appear to demonstrate good performance with acoustic data especially in low SNR regimes, and allow for the possibility of signal reconstruction in highly challenging problems such as gravitational waves (where state-of-the-art methods focus largely on the question of signal detection). 

We demonstrate the possibility of invoking the powerful machinery of TDA for very general classes of signal processing problems, without the requirement of the a-priori presence of prominent geometrical structures in the signals (which is a feature of most classical TDA applications). We believe that this would provide engineers with a new and powerful suite of signal processing tools on one hand, and topological data analysts with a novel and interesting class of problems on the other. Our work opens up the possibility of integrating topology-inspired methods with dedicated state-of-the-art techniques for specific signal processing tasks (such as neural network-based approaches for GW) for improved outcomes, and suggests natural extensions to other time-varying signal classes in statistics, medicine and data science.

\section*{Acknowledgment}

The authors wish to thank Curtis Condon, Ken White, and Al Feng of the Beckman Institute of the University 
of Illinois for the bat data used in Fig. \ref{f:batsig} and for permission to use it in this paper, and Alvin Chua of the National University of Singapore for helpful discussions regarding gravitational wave data. 
JMM and RB were supported by ERC grant \textsc{Blackjack} ERC-2019-STG-851866 and ANR grant \textsc{Baccarat} ANR-20-CHIA-0002. KAT was supported by the NUS Research Scholarship. 
 SSM was partially supported by the INSPIRE research grant DST/INSPIRE/04/2018/002193 from the Dept.~of Science and Technology, Govt.~of India, 
and a Start-Up Grant from Indian Statistical Institute. SG was supported in part by the Singapore MOE grants R-146-000-250-133, R146-000-312-114, A-8002014-00-00 and MOE-T2EP20121-0013.

\newpage
\appendix
\setcounter{equation}{0}

\appendixpage

\section{Tools from topological data analysis} \label{APP:s:tda}
This section gives an in-depth description of the main tools we take from topological data analysis and that we use throughout the main document, namely persistent homology, persistent diagram, minimum volumes and stable volumes.

\subsection{Persistent homology} \label{APP:s:pers_homology}
We extract a few definitions from \cite{EdHa22}, to which we refer for more details. 
Let $S\subset \mathbb{R}^2$ be a finite set of points, and $r>0$. 
Let $B_\bu(r)$ be the closed ball of center $\bu\in\mathbb{R}^2$ and radius $r$.
We consider 
$$ 
    \mathcal{A}(r) = \left \{\sigma\subseteq S: \bigcap_{\bu\in\sigma} (B_\bu(r)\cap V_\bu)\neq \emptyset\right\}, 
$$
where $V_\bu$ is the (closed) cell associated to $\bu$ in the Voronoi diagram of $S$. 
In other words, $V_\bu$ is the set of points in $\mathbb{R}^2$ that are closer to $\bu$ than to any other point in $S$.

For $r>0$, $\mathcal{A}(r)$ is called the $\alpha$\emph{-complex} of $S$ of scale $r$ \cite[Section III.4]{EdHa22}.
The $\alpha$-complex is an example of a \emph{simplicial complex}; Other common choices include the Delaunay and Vietoris-Rips complexes \cite[Chapter III]{EdHa22}.
Elements of any $\mathcal{A}(r)$ are subsets of $S$, and are called \emph{simplices}.
A simplex formed by $p+1$ affinely independent points is called a $p$-simplex, and is identified with its convex hull: a $0$-simplex is thought of as a point, a $1$-simplex as a segment with positive length, a $2$-simplex as a non-flat triangle. $3$-simplices and above do not exist in dimension $2$.
For simplicity, we assume henceforth that no three points of $S$ are collinear, so that all triplets define a non-degenerate triangle.

Figure~\ref{f:alpha_complex} illustrates $\mathcal{A}(r)$ for different values of $r$, with $S$ the set of zeros of a spectrogram.
Momentarily considering $r$ as the \emph{age} of the construction, the basic intuition is that as $r$ grows, holes appear and die in $\mathcal{A}(r)$, and that holes that stay alive for a long time are likely to be indicative of some underlying structure, in our case a signal.
Persistent homology is a set of tools to define and keep track of these holes.

Define $\mathcal{A}(0) = S$.
When $r\ll 1$, $\mathcal{A}(r)= \mathcal{A}(0)$.
As $r$ grows, $\mathcal{A}(r)$ takes on a finite number of values, progressively including all subsets of $S$ at some finite time $r_{\max}>0$.
In short, for any $0<r<s$,
\begin{equation}
    S=\mathcal{A}(0)\subset\mathcal{A}(r)\subsetneq\mathcal{A}(s) \dots \subset \mathcal{A}(r_{\max})=\mathcal{P}(S).
    \eqlabelApp{e:filtration}
\end{equation}
The increasing sequence of sets $\mathcal{A}(r)$, $r\geq 0$, is called a \emph{filtration}. 
Since each complex differs from the previous one by one simplex, we can equivalently describe the filtration \ref{e:filtration} by an ordered list of simplices $\left( (r_1, \sigma_1),\dots, (r_J, \sigma_J) \right)$ for some $J\geq 1$, where $r_{i}$ is the value of $r$ for which $\sigma_{i}$ is added to the filtration.

\begin{figure*}
    \centering
    \hspace*{-5cm}\includegraphics{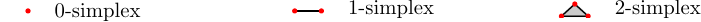} \\

    \begin{subfigure}{0.8\textwidth}
        \hspace{-0.5cm}
        \includegraphics{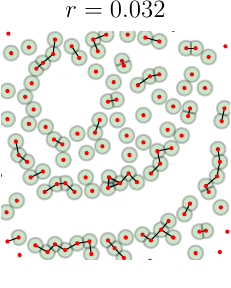}
        \includegraphics{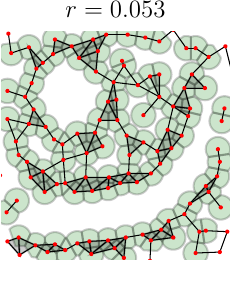}
        \includegraphics{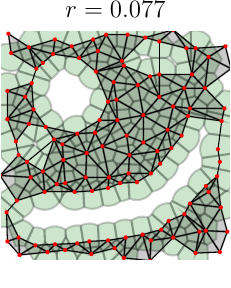}
        \caption{}
        \label{f:alpha_complex}
    \end{subfigure}   
    \hfill
        \begin{subfigure}{0.19\textwidth}
        \hspace*{-0.5cm}
        \includegraphics{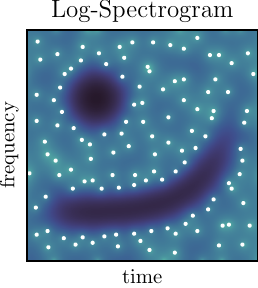}
        \caption{}
        \label{f:alpha_complex_b}
    \end{subfigure}
    \vspace*{1cm}
    
    \begin{subfigure}[b]{0.35\textwidth}
    \hspace*{-1cm}
    \includegraphics{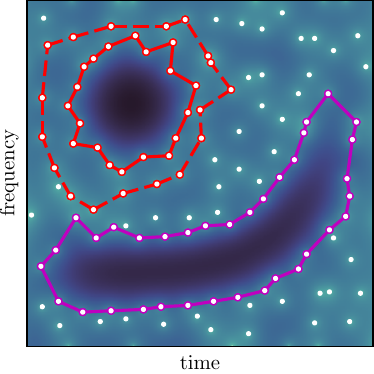}
    \vspace*{0.35cm}
    \caption{Signal+noise and cycles.}
    \label{f:cycles}
    \end{subfigure}
    \hfill
    \begin{subfigure}[b]{0.3\textwidth}
             \includegraphics{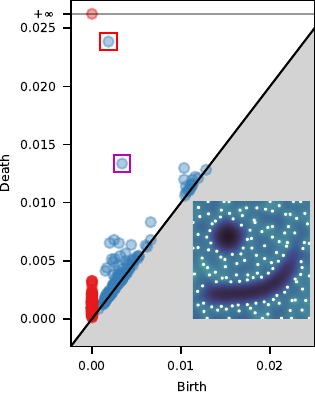} \\
    \caption{Signal+noise.}
    \label{f:pd_b}
    \end{subfigure}
    \hfill
    \begin{subfigure}[b]{0.3\textwidth}
    \hspace{0.5cm}
             \includegraphics{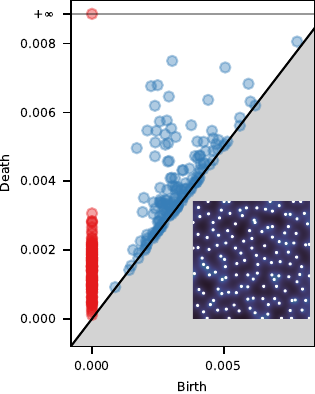} \\ 
    \caption{Noise.}
    \label{f:pd_a}
    \end{subfigure}
    \hfill

    \caption{\textbf{(a)} Simplicial complexes in a filtration for three values of $r$. Red points, black segments, and shaded triangles are, correspondingly, 0-simpleces, 1-simpleces and 2-simpleces. For each point $\bu$, the intersection $B_\bu(r)$ and $ V_\bu$ is shown in green, where $B_\bu(r)$ is a closed ball centered at $\bu$ with radius $r$, and $ V_\bu$ is the Voronoi cell associated to $\bu$; see Sec. \ref{APP:s:pers_homology}. \textbf{(b)} Spectrogram of a mixture of a two-component signal and real white Gaussian noise. The zeros are shown as white dots. \textbf{(c)} Minimum-volume cycles (continuous lines) and a non-minimum cycle (dashed line) superimposed to the log-spectrogram of a signal. \textbf{(d)} and \textbf{(e)} show persistence diagrams $\mathcal{D}_0$ (red dots) and $\mathcal{D}_1$  (blue dots). Each dot correspond to a birth-death pair. \textbf{(d)} Persistence diagrams for a mixture of signal and noise, the spectrogram of which can be seen in the right panel. For $\mathcal{D}_{1}$, the two squared birth-death pairs are significantly further away from the diagonal, corresponding to long-lived \emph{holes} caused by the two-component signal. \textbf{(e)} PD for white Gaussian noise.  }
    \label{f:pd}
\end{figure*}

Denote a formal weighted sum $c= \sum a_i \sigma_i$ of $p$-simplices $\sigma_1, \dots, \sigma_k$, with $a_1, \dots, a_k \in \mathbb{Z}/2\mathbb{Z}$, as a $p$-\emph{chain}, .
Let $r>0$. 
The $p$-chains of $\mathcal{A}(r)$, with the term-by-term sum modulo $2$, form a group\footnote{We reserve sans serif characters for groups.} $\mathsf{C}_p(r)$.
A natural morphism is then the boundary operator $\partial_{p+1}:\mathsf{C}_{p+1}(r)\rightarrow \mathsf{C}_p(r)$, defined as follows.
For a $(p+1)$-simplex $\sigma$, define its \emph{faces} as all its subsets of cardinality $p$. 
The boundary $\partial_{p+1} \sigma$ is then the formal sum of all the faces of $\sigma$, with constant weight $1$.
We then extend $\partial_{p+1}$ to $\mathsf{C}_p(r)$ by linearity, as
$$ 
    \partial_{p+1} \sum a_i\sigma_i := \sum a_i \partial_{p+1}\sigma_i.
$$
The kernel of $\partial_{p}$ is the subgroup $\mathsf{Z}_p(r) < \mathsf{C}_p(r)$ of so-called $p$-\emph{cycles}, and the image of $\partial_{p+1}$ is the subgroup 
$\mathsf{B}_p(r)<\mathsf{C}_p(r)$ of $p$-\emph{boundaries}.
A fundamental result is that $\mathsf{B}_p(r)<\mathsf{Z}_p(r)$, i.e., that a $p$-boundary is a $p$-cycle. 

While introducing algebraic notions might seem an unnecessary overhead at first sight, it crucially allows the translation of topological questions into algebraic manipulations that are amenable to computer implementation. 
In particular, a key notion is the $p$-th homology group $\mathsf{H}_p(r) = \mathsf{Z}_p(r)/\mathsf{B}_p(r)$. 
In words, an element of $\mathsf{H}_p(r)$ (i.e., an equivalence "class" in the quotient) represents a set of $p$-chains obtained by taking a $p$-cycle of $\mathsf{H}_p(r)$, and adding to it any $p$-boundary in $\mathsf{H}_p(r)$. 
For illustration, Figure~\ref{f:holes_and_cycles} gives three cycles in $\mathsf{Z}_1(r)$ for $r$ large enough.
The two cycles in red, which go around the same hole, are in the same equivalence class: one can indeed write the ``larger" red cycle as the sum of the ``smaller" one and the boundary ``in between" the two cycles. 
On the contrary, one cannot obtain the magenta cycle by summing any of the two red cycles and a boundary: the magenta cycle corresponds to a different class in $\mathsf{H}_1(r)$.
Intuitively, equivalence classes in $\mathsf{H}_p(r)$ represent the holes we are trying to detect: for $p=1$, a class contains all loops that encircle a given hole.

Now, as $r$ grows, $\mathsf{H}_p(r)$ gains some classes and loses some, and the \emph{persistence diagram} is a way of keeping track of these events.
Formally, let $r_0>0$.
For $c\in \mathsf{H}_p(r_0)$, denote by $x(c)\in[0,r_0]$ the smallest $r>0$ for which $c\in \mathsf{H}_p(r)$, i.e., $x(c)$ is the birth time of the class.
Similarly, let $y(c)$ be the smallest $r>0$ for which $c\notin \mathsf{H}_p(r)$, i.e., the death time of the class. 
The set of all birth-death time pairs
$$
    \mathcal{D}_p = \{(x(c), y(c)) : \quad c\in \mathsf{H}_p(r) \text{ for some $r>0$}\} \subset \mathbb{R}^2
$$
is called the $p$-th persistence diagram of $S$. 
Plotting the diagrams $\mathcal{D}_0$ and $\mathcal{D}_1$, as in Figure~\ref{f:pd}, gives topological information on the pattern $S\subset\mathbb{R}^2$.
Figure~\ref{f:pd} shows the $0$- and $1$-persistence diagrams corresponding to a real white Gaussian noise realization, whereas Figure~\ref{f:pd_b} shows the diagrams corresponding to the signal with two components used in Figure~\ref{f:alpha_complex}.
On the one hand, $\mathcal{D}_0$ (red dots in Figure~\ref{f:pd}) shows the birth and death times of connected components, all born at $r=0$, and progressively dying until only one remains indefinitely.
On the other hand, $\mathcal{D}_1$ (blue dots in Figure~\ref{f:pd}) gives a visual diagnostic of the number and size of holes in the pattern $S$.
Points far away from the diagonal indicate long-living holes, which we interpret as being indicative of a signal: in Figure~\ref{f:pd_a} two birth-death pairs are significantly further away from the diagonal than the remaining pairs. 
Such pairs are typically representative of signal components.
In contrast, Figure~\ref{f:pd_b} shows that the pairs from $\mathcal{D}_1$ stay close to the diagonal (notice the difference in the scale between the axes in Figure~\ref{f:pd_a} and Figure~\ref{f:pd_b}.

In Section~\ref{s:test} of the main text, we design a detection test based on the persistence diagram, while in Section~\ref{s:signal_estimation} we investigate the possibility to identify the support of the signal once the test has rejected the hypothesis of pure noise. 
For the latter task, we need one more notion from persistence homology. 

\subsection{Minimum volumes} \label{APP:s:minvol_cycles}

Once an anomalous point $(x(c), y(c))$ has been detected in the persistence diagram $\mathcal{D}_1$, the signal processer would like to find a particular cycle corresponding to $c$, and use it to identify a candidate support for the signal that hopefully led to detection.
Since one can add any boundary without changing $c$, it is intuitively appealing to search for a member of $c$ that encloses the smallest possible portion of $\mathbb{R}^2$, that portion likely corresponding to the signal support.
One way to formalize this minimizing class representer is the notion of minimum-volume cycle, introduced in \cite{Sch15} and generalized in \cite{Oba18}.

While the original definition takes the form of the solution to a constrained optimization problem that has to be numerically approximated, our special interest in the alpha-complex and the case of $S\subset\mathbb{R}^2$ allows us to easily obtain minimum-volume cycles through the \emph{persistence tree} (PT) $G=(\mathcal{V},E)$ of $S$ defined by Algorithm~\ref{a:persistence_tree}; the PT construction was introduced by \cite{Sch15}.
 
\begin{algorithm}[h!]
\caption{Building the persistence tree \cite{Sch15,Oba18}}
\label{a:persistence_tree}

\begin{algorithmic}[1]
\Require {A filtration as an ordered list $\left( (r_1, \sigma_1),\dots, (r_J, \sigma_J) \right)$.}
\Ensure The persistence tree $G=(\mathcal{V},E)$.
\State $\mathcal{V} = \{(r_\infty, \sigma_{\infty})\}$, $E = \emptyset$
\For{$j=J,...,1$}
\If{$\sigma_{j}$ is a 2-simplex}
    \State Add node $(r_j,\sigma_{j})$ to $\mathcal{V}$
\ElsIf{$\sigma_{j}$ is a 1-simplex}
    \State Get $\sigma_{s}$ and $\sigma_{t}$, the two 2-simplices sharing $\sigma_{k}$ as an edge.
    \State $\sigma_{s^{\prime}} = \operatorname{Root}(\sigma_{s},\mathcal{V},E)$
    \State $\sigma_{t^{\prime}} = \operatorname{Root}(\sigma_{t},\mathcal{V},E)$
\EndIf

\If{$s^{\prime} == t^{\prime}$}
\State Continue
\ElsIf{$s^{\prime} > t^{\prime}$}
\State Add edge $( \sigma_{s^{\prime}} \overset{r_j} \to \sigma_{t^{\prime}} )$ to $E$. 
\ElsIf{$s^{\prime} < t^{\prime}$}
\State Add edge $( \sigma_{t^{\prime}} \overset{r_j} \to \sigma_{s^{\prime}} )$ to $E$.
\EndIf
\EndFor
\end{algorithmic}
\end{algorithm}

The input to Algorithm \ref{a:persistence_tree} is the list of simplices ordered by appearence in the filtration $\left( (r_1, \sigma_1),\dots, (r_J, \sigma_J) \right)$.
The set of nodes $\mathcal{V}$ is initialized with  $\sigma_\infty= (\mathbb{R}^{2}\cup\{\infty\}) \setminus \mathcal{P}(S)$, which represent the complement of the simplicial complex in the one-point compactification of $\mathbb{R}^2$ \cite{Oba18}.
Each node/edge also contains extra information in $\mathbb{R}$ by saving the value of $r_{i}$ corresponding to $\sigma_{i}$, with $r_{\infty}=\infty$.
Algorithm \ref{a:persistence_tree} simply moves through the filtration in inverse order, starting from $\sigma_{J}$. 
When $\sigma_j$ is a 2-simplex, it is automatically added to $\mathcal{V}$.

If $\sigma_j$ is a 1-simplex, then is added to the list of edges $E$ as long as it joins two unconnected nodes from $\mathcal{V}$.
The $\operatorname{Root}(\sigma_{s},\mathcal{V},E)$ procedure traverses the tree $(\mathcal{V},E)$ upstream, starting from $\sigma_{s}$, searching for a node without parents, i.e. the root corresponding to $\sigma_{s}.$
Note that the notation $\sigma \overset{r_j} \to \sigma_{k}$ indicates the way the persistence tree should be read, meaning here that $\sigma_{k}$ is a descendant of $\sigma$, and the edge $(r_j,\sigma_{j})$ joining both nodes is a 1-simplex $\sigma_{j}$ appearing at $r=r_j$ in the filtration.
\cite{Oba18} shows that $(b,d)\in \mathcal{D}_1(r)$ if and only if there exists $\sigma \overset{r_b} \to \sigma_{d}$, i.e. an edge with $r=r_b$ pointing at a node $\sigma_d$.
The persistence diagram can thus be read from $G$.

Figure~\ref{f:holes_and_cycles_a} shows a simple example of the simplicial complexes obtained for some values of $r$.
The corresponding persistence tree is depicted in Fig.~\ref{f:holes_and_cycles_b}.
In the tree, the 1-simplex $\sigma_{b}=[0,4]$, corresponding to the birth of a non-trivial cycle at $r=0.097$, points to the 2-simplex $\sigma_{d}=[0,1,4]$, the appearance of which marks the death of this homological class at $r=0.152$.
Formally, a minimum-volume cycle generating the pair $(b,d) \in \mathcal{D}_1$ is defined as $\partial_2 MV(\sigma_d)$, where $MV(\sigma_d)$ is the volume obtained as
\begin{equation}\eqlabelApp{e:min_vols}
    MV(\sigma_d) = \{ \sigma_d\} \cup \operatorname{des}(\sigma_d,E),
\end{equation}
$\operatorname{des}(\sigma_d,E)$ being the descendants of the node $\sigma_d$ in the persistence tree $G$ (dashed blue line in the PT of Fig. \ref{f:holes_and_cycles_b}).
As shown in \cite{Oba18}, $MV(\sigma_d)$ is the volume with the smallest size, measured in number of 2-simplices, for the component associated with a given $(b,d) \in \mathcal{D}_1$.
Hence, we will refer to \eqref{e:min_vols} as a \emph{minimum} volume (MV).

Fig.~\ref{f:holes_and_cycles_b}, the MV corresponding to $\sigma_d=[0,1,4]$ is displayed in blue in the bottom section of the right panel.
This notion corresponds to the intuition conveyed by Figure~\ref{f:holes_and_cycles_a}: at the iteration where the class died, all the volume is then ``filled''.
We can recover the MV by tracking the 2-simplices in it in the reverse order of appearance in the filtration, starting from $\sigma_{d}$ and adding one 2-simplex $\sigma_i$ at the time for $r_i > r_b$.
This is exactly what we do by taking the descendants of $\sigma_d$ from the persistence tree in \eqref{e:min_vols}.

\begin{figure*}
    \centering
    \begin{subfigure}{0.46\textwidth}
    \centering
        \includegraphics[width=\textwidth]{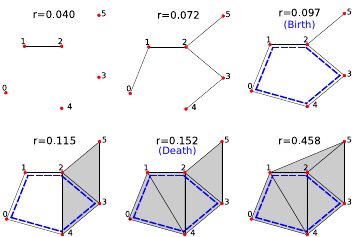}
        \caption{Filtration.}\label{f:holes_and_cycles_a}
    \end{subfigure}
    \begin{subfigure}{0.46\textwidth}
    \centering
        \vspace*{-1.5cm}
        \includegraphics[width=\textwidth]{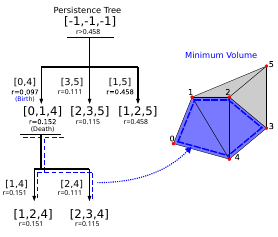}
        \caption{Persistence tree.}\label{f:holes_and_cycles_b}
    \end{subfigure}
    \\
    \interspace
    \begin{subfigure}[b]{0.45\textwidth}
    \includegraphics[width=\textwidth]{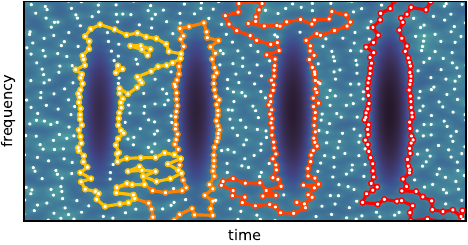} 
    \caption{Impulses. Minimum volumes.} \label{f:stable_volumes_a}
    \end{subfigure}
    \hfill
    \begin{subfigure}[b]{0.45\textwidth}
    \includegraphics[width=\textwidth]{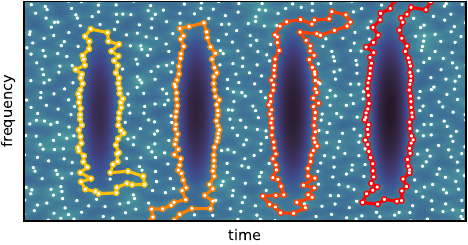}
    \caption{Impulses. Stable volumes.} \label{f:stable_volumes_b}
    \end{subfigure}
    \\
    \vspace*{1cm}
    \begin{subfigure}[b]{0.6\textwidth}
    \centering
    \includegraphics[width=\textwidth]{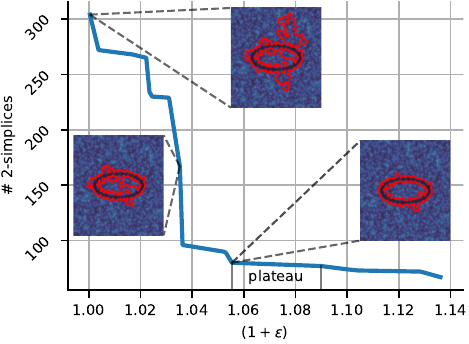}
    \caption{Volume size (in number of 2-simplices) vs. bandwidth parameter $\epsilon$.}
    \label{f:tunning_epsilon}
    \end{subfigure}

    \caption{\textbf{(a)} Alpha-complexes in a filtration, for different values of $r$. The minimum-volume cycle born at $r=0.097$ is shown in blue, dashed line. \textbf{(b)} Persistence tree corresponding to the filtration (top), where $[-1,-1,-1]$ represents $\sigma_{\infty}$, and the minimum volume (bottom, in blue) corresponding to the class born at $r=0.097$. The blue, dashed line in the persistence tree links the 2-simplices comprised in the minimum volume. \textbf{(c)} and \textbf{(d)} show minimum volumes and stable volumes, respectively, for comparison. \textbf{(e)} Shows the volume size of the stable volumes for different values of $\epsilon$.}
    \label{f:holes_and_cycles}
\end{figure*}

Returning now to an example based on the spectrogram, Fig. \ref{f:cycles} illustrates these concepts by superimposing different cycles on top of a spectrogram. 
The continuous red and magenta lines show the minimum-volume cycles corresponding to the significant pairs $(x(c), y(c))$ highlighted in Fig. \ref{f:pd_b} (the color indicates the corresponding pair in the persistence diagram).
Both red lines, continuous and dashed, show two cycles in the same class, i.e. represented by the same pair in the persistence diagram. 
However, the dashed cycle is not a minimum-volume cycle.

Finally, note that if one takes two points in the persistence diagram, nothing prevents the two corresponding minimum volumes to intersect. 
In our particular application of identifying signal components in the time-frequency plane, we prefer restricting ourselves to non-overlapping volumes.
To achieve that, we simply restrict to minimum volumes \eqref{e:min_vols} such that $\sigma_d$ is a child of the root $\sigma_{\infty}$.
By construction, none of these minimum volumes intersect.

\section{Stable volumes} \label{APP:s:stable_vols}
We noticed that a small change in the input signal, like the difference between two noise realizations, can produce large variations of the cycles --and their volumes--, even for the cycles associated to long-living components \cite{cohen2005stability,Ben20,Oba23}.  
Such instability in the presence of noise is a known issue in computational topology, and some approaches to \emph{stabilize} the components of interest have been proposed \cite{Ben20,Oba23}.

Intuitively, one way to robustify the construction of minimum volumes is to take the intersection of minimum volumes corresponding to slightly altered inputs, where the alteration is small enough that we can identify a volume from one signal to the other. Formally, we focus here on the so-called \emph{stable volumes} recently proposed by Obayashi \cite{Oba23}, which can be defined directly from the persistence tree in Sec. \ref{APP:s:minvol_cycles}.

More precisely, given a filtration $\left( (r_1, \sigma_1),\dots, (r_J, \sigma_J) \right)$, a pair $(b,d)\in\mathcal{D}_1$ and its corresponding persistence tree, the stable volume with noise bandwidth $\epsilon\geq 0$ is defined as
\begin{equation} \eqlabelApp{e:stable_vols}
    SV_{\epsilon}(\sigma_b,\sigma_d) = \{ \sigma_d\} \cup \left( \bigcup\limits_{\sigma \in C_{\epsilon}(\sigma_b,\sigma_d)} \operatorname{des}(\sigma,G)  \right),
\end{equation}
where $C_{\epsilon}(\sigma_b,\sigma_d) = \{\sigma \in X^{(2)} \;|\; \sigma_d \xrightarrow{r_k} \sigma \text{ and } r_k \geq (1+\epsilon) \; r_b \}$, $X^{(2)}$ is the set of all 2-simplices.
Obayashi \cite{Oba23} shows that \eqref{e:stable_vols}
is indeed the intersection of minimum volumes corresponding to a certain sophisticated type of perturbation of the filtration, which is a trade-off between actually representing noise in the input and allowing the identification of minimum volumes across perturbations.

By comparing \eqref{e:stable_vols} to \eqref{e:min_vols}, it can be interpreted that in order to get a stable volume, one needs first to prune the branches stemming from the node $\sigma_d$ in $G$ for $r<(1+\epsilon)r_b$ in the filtration.
Then the stable volume is computed by aggregating the remaining descendants of $\sigma_d$, in a similar fashion as the minimum volumes defined in Sec. \ref{APP:s:minvol_cycles}.

Figures \ref{f:stable_volumes_a} and \ref{f:stable_volumes_b} illustrate the difference between minimum and stable volumes, obtained for the same signal.
By comparing both cases, one can see that stable volumes (Fig. \ref{f:stable_volumes_b}) produce a more accurate estimation of the region of interest, in this case the domain corresponding to each of the four impulses in the signal.

The bandwidth parameter $\epsilon$ plays a key role in \emph{stabilizing} a minimum-volume cycle, since $\epsilon=0$ makes \eqref{e:stable_vols} equivalent to \eqref{e:min_vols}. 
Using $\epsilon>1.0$ produces a stable volume the size of which --in number of 2-simplices-- is smaller than the minimum-volume cycle obtained for the same point in the PD.

Figure \ref{f:tunning_epsilon} shows the monotonically decreasing relationship between the size of the volume and $\epsilon$ for a signal with $\text{SNR}=0$ dB, i.e. noise and signal with the same energy.

Theoretically, $\epsilon$ depends on the SNR of the signal \cite{Oba23}. 
However, an explicit relationship between the bandwidth parameter and the level of noise in the signal is not available in this case.
Hence, in practice, $\epsilon$ should be large enough to reduce the instabilities, but not as large that the resulting volume is too small to correctly represent the signal domain of interest.
The next section describes a practical approach to tune $\epsilon$.

\subsection{Tuning the bandwidth parameter} \label{APP:s:tuning_sv}
As suggested by Obayashi in \cite{Oba23}, a possible heuristic to fix the bandwidth parameter consists in counting the number of 2-simplices in the volume obtained for a range of values of $\epsilon$ in order to build a curve like the one shown in Fig. \ref{f:tunning_epsilon}.
A plateau in the curve indicates a stabilization of the volume, hence a value of $\epsilon$ in a plateau of this curve should be selected. 
For each component extracted from the PD, we look for a plateau in an interval $[0, \epsilon_{\max}]$, and the minimum value of $\epsilon$ in the plateau is chosen.
An $\epsilon_{\max}=0.15$ was empirically chosen. 

It might occur that there are not enough points in the interval $[0, \epsilon_{\max}]$ to form a plateau.
When there are less than ten points in the search interval $[0, \epsilon_{\max}]$, the latter is expanded until a plateau can be computed. 
This approach is used to automatically fix $\epsilon$ when stable volumes are computed, as in Fig. \ref{f:stable_volumes_b}.

\subsection{Signal Reconstruction}
Algorithm \ref{alg:sig_est} summarizes the steps for signal estimation based on minimum (or stable) volumes.
\begin{algorithm} [h!]
\caption{Signal reconstruction}
\label{alg:sig_est}
\begin{algorithmic}[1]
\Require Noisy signal $h$, significance level $\alpha$.
\State Compute the STFT $V_{g}(h)$ using \eqref{e:stft}.
\State Find the zeros of $|V_{g}(h)|^2$.
\State Compute the persistence diagram based on the zeros.
\State Estimate the number of significant components to extract from the PD using the tests described in Sec. \ref{s:number_holes} with significance level $\alpha$.
\State Compute the volume of the minimum-volumes cycle (or stable volume) associated with each relevant component.
\State Approximate the signal domain $D$ as the union of the computed volumes.
\State Compute a signal estimation $\tilde{f}$ using \eqref{e:inv_stft} of the main document.
\State \Return $\tilde{f}$.
\end{algorithmic}
\end{algorithm}

The performance of the reconstruction for the synthetic signals was measured using the quality reconstruction factor (QRF, see \cite{meignen2016adaptive}) given by $ \operatorname{QRF} = 10 \log_{10}\left( \Vert f \Vert ^2_{2}/\Vert f-\tilde{f} \Vert^{2}_{2} \right)\; \text{(dB)}$, 
where $f$ is the original noiseless signal and $\tilde{f}$ is a \emph{denoised} approximation.
The QRF can be understood as an estimation of the SNR of the output signal after a reconstruction approach such as Algorithm \ref{alg:sig_est}.
In practice, it should be higher than the SNR of the noisy input signal to indicate an improvement.
Notice that one can only compute the QRF if the noiseless reference is available, hence the importance of working with synthetic signals at an early stage.

\section{Hypothesis testing} \label{APP:s:hyp_test}
This section complements the contents in the main document, providing more technical details regarding the statistical tools we use for signal detection and denoising.

\subsection{Simple hypothesis} \label{APP:s:hyp_test:simple}
Given a noise model $\xi$ and a candidate signal $f$, signal detection can be stated as a binary hypothesis test \cite{flandrin1998time,BaFlCh18}:
\begin{equation} \eqlabelApp{APP:e:test}
 \left\lbrace \begin{array}{l}
               H_{0}: h = \xi \;,\text{ i.e. the observation is pure noise} \\
               H_{1}: h = f + \xi\;,\text{ i.e. the observation is a mixture of signal plus noise} 
              \end{array}\right. .
\end{equation}
Formally, both $H_0$ and $H_1$ here can be identified with singletons containing one probability measure each. 
Alternately, if no candidate signal $f$ is known, one can define $H_1$ as the complement $H_0^C$ of $H_1$.
A test aiming to detect a signal can be characterized by its
\begin{enumerate}
    \item Specificity: the probability $1-\alpha$ of correctly keeping $H_0$, where $\alpha$ is the so-called \emph{significance} of the test, i.e. the probability of wrongfully rejecting $H_0$.
    \item Sensitivity: the probability of correctly rejecting $H_0$, also known as the \emph{statistical power of the test}.
\end{enumerate}

The main idea behind our detection tests is computing the top order statistics $\Xi_{(\ell)},\; \ell=1,2,\dots,$ for a generic test statistic $\Xi$. In our case, $\Xi$ could be the distance to the diagonal in the PD, or the \emph{energy} associated to the minimum (or stable) volume of that component in the spectrogram.
One should be able to detect the presence of a signal by comparing the quantiles of the top order statistics against the same quantiles under the null model (e.g., white Gaussian noise). For practical purposes, quantiles are obtained by Monte Carlo simulations. This procedure is explained below.

From the underlying noise model, we generate $B$ independent samples. 
Let $\Xi_{(1), 1}, \ldots, \Xi_{(1), B}$ denote the corresponding top order statistics. 
Then a $p$-value for testing $H_0$ can simply be obtained as 
\[\frac{1}{B} \sum_{j = 1}^B \mathbb{I}(\Xi_{(1), j} > \Xi_{(1), \mathrm{obs}}),\]
where $\Xi_{(1),\mathrm{obs}}$ is the top order statistics obtained from the sample to be tested. 
A signal is then \emph{detected} if the $p$-value is below a predefined value of $\alpha\in(0,1)$.
By a symmetry argument, if $B\alpha$ is an integer, the significance of the test is guaranteed to be at least $\alpha$.

\subsection{Multiple hypothesis testing} 
\label{APP:s:hyp_test:mult}
Our applications on signal detection and reconstruction require inferring the existence of one or more signal components.
This is addressed by using multiple hypothesis testing, conducting several tests simultaneously (for signal detection) or sequentially (for estimating the number of signal components).

\subsubsection{Family-wise error rate}
A particularly important aspect of multiple testing is to construct procedures capable of controlling a chosen type I error-related criterion.
A traditional choice is the Family-wise Error Rate (FWER), which is the worst-case probability of incorrectly rejecting any null-hypothesis.
Formally, this can be defined as follows; see e.g. \cite{fromont2016family}.
Let $(\mathbb{X},\mathcal{X})$ be a mesurable space, and $\mathcal{P}$ the set of probability measures on $\mathbb{X}$.
A hypothesis $H$ can be seen as a subset of $\mathcal{P}$, and $H$ is true under $P$ if $P\in H$, and false if $P\notin H$.
Denote the set of true hypotheses under $P$ as $\mathcal{T}(P)=\{H\in\mathcal{H}: P\in H\}$.
Then, writing $\mathcal{R}$ for the (observation-dependent) subset of $\mathcal{P}$ that is rejected,
the FWER is defined as \[ \text{FWER}(\mathcal{R}) = \sup_{P\in\mathcal{P}} P \left(\mathcal{R}\cap \mathcal{T}(P) \neq \emptyset\right). \]

\subsubsection{Simultaneous tests} \label{APP:s:hyp_test:mult:simul}
With the notations of Section~\ref{APP:s:hyp_test:simple}, we may want to base our decision on several top order statistics $\Xi_{(k)}$ instead of a single one. 
This leads to multiple hypothesis testing, which we now describe more formally.
Start with a suitable choice of $K$ most persistent components to evaluate, i.e. the components associated with the $K$ longest distances to the diagonal in the PD.
Denote by $\mu_k^{(0)} $, $1\le k\le K$, the marginal distribution of $\Xi_{(k)}$ under our null hypothesis $H_0$ that the original signal is pure white Gaussian noise.
Then, we conduct $K$ tests simultaneously for the $K$ hypotheses
\[
    H_{0,k} : \Xi_{(k)} \sim \mu_k^{(0)}, \quad k=1, \dots, K.
\]
Formally, $H_{0,k}$ is the set of distributions over signals under which the marginal of $\Xi_{(k)}$ is $\mu_k^{(0)}$.

Again, from the $B$-many simulated persistence diagrams, with statistics $\Xi_{(k), j}, 1 \le j \le B$, we compute a series of $p$-values $p_1, p_2, \dots, p_{K}$ as \[p_k := \frac{1}{B}\sum_{j=1}^B \mathbb{I}(\Xi_{(k),j} > \Xi_{(k), \mathrm{obs}}).\]
For each $H_{0,k}$, we perform a level-$\alpha_k$ test, for some $\alpha_k\in (0,1)$. 
We define the hypothesis of pure noise to be rejected if and only if $H_{0,k}$ is rejected for at least one $1\le k\le K$. 
For any distribution $P$ under which data and the $B$ simulations are i.i.d., 
\begin{equation}
    \eqlabelApp{e:almost_FWER}
    P(H_{0,k}\text{ is rejected for at least one $k$}) = P\big(\cup_{k=1}^K \{p_k < \alpha_k \}\big) < \sum_{k=1}^K \alpha_k.
\end{equation}
Choosing 
\begin{equation}
    \eqlabelApp{e:bonferroni}
    \alpha_k = \alpha/K
\end{equation}
for some predefined $\alpha\in(0,1)$, \eqref{e:almost_FWER} and thus the FWER are bounded by $\alpha$. 
The choice of \eqref{e:bonferroni} is called Bonferroni's correction, and is often the default choice in practice among so-called \emph{$\alpha$-corrections}.

As an aside, note that multiple hypothesis testing can be used to estimate the number of holes, e.g. using as estimator the cardinality of $\{1\le k\le K:p_k < \alpha_k\}$. 
However, we found this estimator not to be robust to the presence of spurious holes in noise, and rather recommend sequential testing.

\subsubsection{Sequential tests} \label{APP:s:hyp_test:mult:seq}

In Section~\ref{APP:s:hyp_test:mult:simul}, we claimed a rejection when at least one --not necessarily the first-- of the hypotheses $H_{0,1}, \dots, H_{0,K}$ was violated.
Instead, we could check whether $p_1 < \alpha_1$, and if not then check whether $p_2 < \alpha_2$, and continue so until either $H_{0,k}$ is rejected or we reach the final $K$th test and do not reject. 
This leads to another rejection region, i.e. another test, which we call a \emph{sequential test}.
This also naturally suggests another estimator for the number of holes.
Indeed, intuitively, rejecting $H_{0,k}$ for $k =1, \dots, k_{\star}$ is an indication of the presence of at least $k_{\star}$ holes in the configuration of zeros. 
Formally, we set
$$
    \hat{n}_{\mathrm{holes}} = \min \{0 \le k \le K-1 : p_{k+1} \geq \alpha_{k+1}\},
$$
with the convention that $\min\emptyset = 0$.
The power of the sequential test, as well as the distribution of the estimator $\hat{n}_{\mathrm{holes}}$, depend on the choice of $(\alpha_k)$.

\subsubsection{Choice of levels of significance in sequential tests} \label{APP:s:hyp_test:mult:seq:alpha_levels}

For any distribution $P$ under which data and the $B$ simulations are i.i.d., the test wrongly detects if and only if $\hat{n}_{\text{holes}}>0$, so that
$$
    P(\text{false detection}) = P(\hat{n}_{\text{holes}}>0) = P(p_1\ge\alpha_1)=\alpha_1.
$$
So, as long as $\alpha_1\leq \alpha$, any 
$\alpha$-correction controls the FWER of the sequential test.
We consider three different choices of $\alpha$-corrections $(\alpha_k)_{1\leq k\leq K}$ in the main paper: the Bonferroni correction we used in simultaneous testing \eqref{e:bonferroni}, as well as 
\begin{enumerate}
    \item Polynomial decay: $\alpha_k = \frac{\alpha}{k^{m}},\quad m \ge 0$.
    \item Geometric decay: $\alpha_k = \alpha \beta^{k}, \quad \beta \in (0, 1)$.
\end{enumerate}
Note that for the sequential test, the Bonferroni correction is overly conservative.

Finally, on top of controlling the FWER, we note that fixing an $\alpha$-correction allows to further elaborate on the distribution of the estimator $\hat{n}_{\text{holes}}$. 
For instance, let $k^\star>0$ and assume that $P\in H_{0,1}^C\cap H_{0,2}^C\cap\dots\cap H_{0,k^*}^C\cap H_{0,k^*+1}$.
Intuitively, this translates to the fact that $k^*$ is the underlying number of holes in the signal. 
Let us compute the probability under such a $P$ of misestimating $k^\star$.

\begin{proposition}
    Assume that $P\in H_{0,1}^C\cap H_{0,2}^C\cap\dots\cap H_{0,k^*}^C\cap H_{0,k^*+1}$. Let
    \[
        \tilde{\alpha}_l := P(p_l< \alpha_l\,|\,p_j<\alpha_j \text{ for }j < l)
    \]
    and
    \[
        \epsilon_l := 1 - \tilde{\alpha}_l =  P(p_l\ge \alpha_l\,|\,p_j<\alpha_j \text{ for } j < l).
    \]
    Then
    \begin{align} \eqlabelApp{APP:e:underestim}
        P(\hat{n}_{\text{holes}}<k^*)=\sum_{j=0}^{k^*-1}\bigg(\epsilon_{j+1}\prod_{l=1}^j(1-\epsilon_l)\bigg)
    \end{align}
    and
    \begin{align} \eqlabelApp{APP:e:overestim}
        P(\hat{n}_{\text{holes}} > k^*)=\bigg(\prod_{l=1}^{k^*}(1-\epsilon_l)\bigg)\cdot\sum_{j >k^*}\bigg((1-\tilde{\alpha}_{j+1}) \prod_{l=k^*+1}^j\tilde{\alpha}_l\bigg).
    \end{align}
\end{proposition}
\begin{proof}
For $j<k^*$, 
\begin{align*}
    P(\hat{n}_{\text{holes}}=j) = P(p_1<\alpha_1, p_2<\alpha_2,\dots, p_{j+1}\ge \alpha_{j+1})=\epsilon_{j+1}\prod_{l=1}^j(1-\epsilon_l),
\end{align*}
and hence
\[
    P(\hat{n}_{\text{holes}}<k^*)=\sum_{j=0}^{k^*-1}\bigg(\epsilon_{j+1}\prod_{l=1}^j(1-\epsilon_l)\bigg).
\]
On the other hand, for $j>k^*$, 
\begin{align}
P(\hat{n}_{\text{holes}}=j) &=  P(p_1<\alpha_1, p_2<\alpha_2,\dots, p_{j+1}\ge \alpha_{j+1}) \\
&= \bigg(\prod_{l=1}^{k^*}(1-\epsilon_l)\bigg)\cdot(1-\tilde{\alpha}_{j+1}) \prod_{l=k^*+1}^j\tilde{\alpha}_l. \nonumber
\end{align} 
Therefore
\[
    P(\hat{n}_{\text{holes}} > k^*)=\bigg(\prod_{l=1}^{k^*}(1-\epsilon_l)\bigg)\cdot\sum_{j >k^*}\bigg((1-\tilde{\alpha}_{j+1}) \prod_{l=k^*+1}^j\tilde{\alpha}_l\bigg).
\]
This completes the proof.
\end{proof}
For a good test, we will have $\epsilon_l$ small for all $l \le k^*$ and $\tilde{\alpha}_l$ controlled for all $l > k^*$. As such, it is clear that the probability of underestimating $k^*$ would be small as long as the tests we use have good power. We now look at the probability of overestimating $k^*$ more carefully.
\begin{lemma}
    Suppose that $\alpha_0 \ge \alpha_1 \ge \alpha_2 \ge \cdots$ and that for any $l > k^*$, we have $\tilde{\alpha}_l \le \alpha_l (1 + \delta)$ for some $\delta \ge 0$ such that $\alpha_0 (1 + \delta) < 1$. Then
    \[
        P(\hat{n}_{\text{holes}} > k^*) \le \frac{\alpha_{k^* + 1} \big(\prod_{l=1}^{k^*}(1-\epsilon_l)\big)(1 + \delta)}{1 - \alpha_0(1 + \delta)}.
    \] 
\end{lemma}
\begin{proof}
We have
\begin{align*}
    \sum_{j >k^*} \bigg((1-\tilde{\alpha}_{j+1}) \prod_{l=k^*+1}^j\tilde{\alpha}_l\bigg) &\le \sum_{j >k^*} \prod_{l=k^*+1}^j\alpha_l(1 + \delta) \\
    &= \sum_{j >k^*} (\alpha_{k^* + 1} (1 + \delta))^{j - k^*} \\
    &= \frac{\alpha_{k^* + 1}(1 + \delta)}{1 - \alpha_{k^* + 1}(1 + \delta)},
\end{align*}
whence the desired bound follows.
\end{proof}

Let us now assume that for $l > k^*$,
\begin{equation}\label{eq:conditional-level-assmumption}
    \alpha_l (1 - \delta) \le \tilde{\alpha}_l \le \alpha_l (1 + \delta)
\end{equation}
for some $0 \le \delta \le \delta_0$, where $\delta_0 < 1$ is such that $\alpha_0 (1 + \delta_0) < 1$. If the $p$-values are independent, then one may take $\delta = 0$. If the degree of dependence among the $p$-values is small, then a small choice of $\delta$ would work. Assuming \eqref{eq:conditional-level-assmumption}, we now give explicit formulas for the overestimation probability $P(\hat{n}_{\text{holes}} > k^*)$.

\begin{lemma}
    Suppose that $\alpha_0 \ge \alpha_1 \ge \alpha_2 \ge \cdots$ and that \eqref{eq:conditional-level-assmumption} holds. Then
    \[
        P(\hat{n}_{\text{holes}} > k^*) = \alpha_{k^* + 1} \bigg(\prod_{l=1}^{k^*}(1-\epsilon_l)\bigg)(1 + O(\delta)),
    \]
    where the implicit constant in $O(\cdot)$ depends on $\alpha_0$ and $\delta_0$.
\end{lemma}
\begin{proof}
We have
\begin{align*}
    \sum_{j >k^*}&\bigg((1-\tilde{\alpha}_{j+1}) \prod_{l=k^*+1}^j\tilde{\alpha}_l\bigg) \\
    &\le \sum_{j >k^*}\bigg((1-\alpha_{j+1}(1 - \delta)) \prod_{l=k^*+1}^j\alpha_l(1 + \delta)\bigg) \\
    &= \sum_{j >k^*}\bigg((1-\alpha_{j+1}(1 + \delta)) \prod_{l=k^*+1}^j\alpha_l(1 + \delta)\bigg) + \frac{2 \delta}{1 + \delta} \sum_{j > k^*}\prod_{l=k^*+1}^{j + 1}\alpha_l(1 + \delta) \\
    &=: T_1 + T_2.
\end{align*}
Note that $T_1$ is a telescopic sum:
\[
    T_1 = \sum_{j>k^*}\bigg(\prod_{l = k^* + 1}^j \alpha_l (1 + \delta) - \prod_{l = k^* + 1}^{j + 1} \alpha_l (1 + \delta)\bigg) = \alpha_{k^* + 1} (1 + \delta).
\]
On the other hand,
\begin{align*}
    T_2 \le \frac{2 \delta}{1 + \delta} \sum_{j > k^*}(\alpha_{k^* + 1} (1 + \delta))^{j - k^* + 1} = \frac{2 \delta}{1 + \delta} \cdot \frac{(\alpha_{k^* + 1} (1 + \delta))^2}{1 - \alpha_{k^* + 1} (1 + \delta)}.
\end{align*}
We have thus shown that
\begin{align*}
    \sum_{j >k^*}\bigg((1-\tilde{\alpha}_{j+1}) \prod_{l=k^*+1}^j\tilde{\alpha}_l\bigg) &\le \alpha_{k^* + 1}(1 + \delta) \bigg[1 + \frac{2 \delta}{1 + \delta} \cdot \frac{\alpha_{k^* + 1} (1 + \delta)}{1 - \alpha_{k^* + 1} (1 + \delta)}\bigg] \\
    &= \alpha_{k^* + 1}(1 + \delta) \bigg[1 + \frac{{2\alpha_0(1 + \delta_0)}}{1 - \alpha_0 (1 + \delta_0)}\delta\bigg] \\
    &= \alpha_{k^* + 1}(1 + O(\delta)).
\end{align*}
Starting with the lower bound
\begin{align*}
    \sum_{j >k^*}&\bigg((1-\tilde{\alpha}_{j+1}) \prod_{l=k^*+1}^j\tilde{\alpha}_l\bigg) \\
    &\ge \sum_{j >k^*}\bigg((1-\alpha_{j+1}(1 - \delta)) \prod_{l=k^*+1}^j\alpha_l(1 - \delta)\bigg) - \frac{2 \delta}{1 - \delta} \sum_{j > k^*}\prod_{l=k^*+1}^{j + 1}\alpha_l(1 - \delta),
\end{align*}
and proceeding similarly, one may prove that
\begin{align*}
    \sum_{j >k^*}\bigg((1-\tilde{\alpha}_{j+1}) \prod_{l=k^*+1}^j\tilde{\alpha}_l\bigg) &\ge \alpha_{k^* + 1} (1 - \delta) \bigg[1 - \frac{2 \delta}{1 - \delta} \cdot \frac{\alpha_{k^* + 1} (1 - \delta)}{1 - \alpha_{k^* + 1} (1 - \delta)}\bigg] \\
    &\ge \alpha_{k^* + 1} (1 - \delta) \bigg[1 - \frac{2\alpha_0}{(1 - \delta_0) (1 - \alpha_0)} \delta\bigg] \\
    &=\alpha_{k^* + 1}(1 - O(\delta)).
\end{align*}
This completes the proof.
\end{proof}

\begin{corollary}
Under the polynomial decay $\alpha$-correction, i.e. with $\alpha_k=\frac{\alpha}{k^m}$, $m\ge 0$, assuming that \eqref{eq:conditional-level-assmumption} holds, we have
\begin{align*}
    P(\hat{n}_{\text{holes}} > k^*) &= \frac{\alpha}{(k^*+1)^m}\bigg(\prod_{l=1}^{k^*}(1-\epsilon_l)\bigg)(1 + O(\delta)).
\end{align*}
Similarly, under the geometric decay $\alpha$-correction, i.e. with $\alpha_k=\alpha\beta^{k}\,,\,\beta\in(0,1)$, assuming that \eqref{eq:conditional-level-assmumption} holds, we have
\begin{align*}
    P(\hat{n}_{\text{holes}} > k^*) &= \alpha \beta^{k^*+1}\bigg(\prod_{l=1}^{k^*}(1-\epsilon_l)\bigg)(1 + O(\delta)).
\end{align*}
\end{corollary}


In either $\alpha$-correction method, the probability of inaccurately estimating the number of holes depends on the $\epsilon_j$'s and $\delta$, which in turn depend on the choice of test statistics $\Xi_{(j)}$ and their distributions $\mu_j^{(0)}$. 
A good choice of test statistics $\Xi_{(j)}$ would be able to distinguish between noise and true signal, ensuring $\epsilon_j$ remains small for all $j \le k^*$. 
Ideally, if the dependence among the order statistics $\Xi_{(j)}$ (for $1\le k\le K$) is negligible, then $\delta$ will also be negligible.
Coupled with a suitable choice of $\alpha, m$ and $\beta$, both polynomial decay and geometric decay $\alpha$-corrections will estimate the exact number of holes correctly with a large probability.

\section{Detection based on the accumulated persistence function} \label{APP:APF}
A common variant of the persistence diagram is the \emph{rotated-rescaled persistence diagram} (RRPD).
A birth-death pair $(b_i, d_i)$ appears in the RRPD diagram as coordinates $(m_i,l_i)$, where $m_i=(b_i+d_i)/2$ and $l_i = d_i - b_i$, i.e. the half-life and total life of a component.
The Accumulated Persistence Function (APF) \cite{biscio2019accumulated} is defined as a function of the half-life 
\begin{equation}\label{e:apf}
    \APF :m \mapsto \sum\limits_{i} l_{i} \mathbb{I}(m_i\leq m),
\end{equation}
where we assume that the RRPD is computed from $\mathcal{D}_{0}$ or $\mathcal{D}_{1}$, $m\in[0,m_{\max}]$. 
Note that $\APF$
is a monotonically increasing function.
Following \cite{baddeley2014tests, myllymaki2017global}, a test can be devised by first choosing a significance level $\alpha$, simulating $m$ signals coming from the null distribution, computing
\begin{equation}
    t_{\ell} = \max\limits_{m\in[0,m_{\max}]}\left(\APF_{\ell}(m) - \overline{\APF}(m) \right), \quad \ell=0,\dots,L,
\end{equation}
where $\APF_0$ is the APF obtained from the observed data, $\APF_\ell$ the APF obtained from the $\ell$th simulated data ($1\leq \ell\leq L)$, and $\overline{\APF}$ is the average of all $\APF_{\ell}$, $\ell=0,\dots,L$.
Finally, the latter are then sorted in ascending order $t_{(\ell)},\ell=1,\dots,L$, and the null hypothesis is rejected if $t_{0}\geq t_{(\ell^{*})}$, where $\ell^{*}=\alpha(L+1)$.
Notice that $\alpha$ and $L$ should be chosen so that $\ell^{*}$ is an integer. 
By a symmetry argument, the resulting test has significance $\alpha$ \cite{baddeley2014tests}.

\section{Detection of gravitational waves} \label{APP:s:gw_detection}
Using persistent homology to detect gravitational waves is particularly intriguing, as it differs significantly from applications involving synthetic signals with background white noise. 
In this section, we provide more details regarding the GW application described in the main document.

\subsection{The Kaggle Gravitational Waves Detection dataset}
For illustration purposes, we use the gravitational wave data from the Kaggle Gravitational Waves Detection competition \cite{kaggle-g2net}. 
The dataset consists of realistic simulations of gravitational wave measurements from a network of three gravitational wave interferometers (LIGO Hanford, LIGO Livingston, and VIRGO).
Each data sample, or \emph{event}, contains three time series, each 2 seconds long and sampled at $2048$ samples per second. 
The data is labeled as ``$0$'' if there is only noise and ``$1$'' if there is a noisy gravitational wave signal.
The detection task is challenging, particularly for data with very low SNR. 

\subsubsection{Data subsets}
In order to characterize in depth the performance of the signal detection tests for gravitational waves, we divided the original dataset in different groups according to their SNR.
Since an estimation of the SNR of each signal is not available, we used an adaptive hard threshold \cite{donoho1994ideal,mallat2008wavelet} on the spectrogram to segment the dataset first, following the notion that when the signature of a gravitational wave in the spectrogram is clear its values will surpass a higher threshold than when the signal is deeply buried in noise.
Given that the real and imaginary parts of the STFT of white Gaussian noise $\xi$ with variance $\sigma^2$ are normally distributed, and provided that the analysis window in \eqref{e:stft} is normalized to $\Vert g \Vert^{2}_{2}=1$, the spectrogram ordinates follow an exponential distribution $\text{Exp}(\lambda): t\mapsto\frac{1}{\lambda}e^{\frac{t}{\lambda}}$ with parameter $\lambda = \sigma^2$ \cite{pham2018novel, flandrin2018explorations}.

We define the normalized spectrogram maximum (NSM), $\hat{\zeta} = \hat{\gamma} / \hat{\sigma}$ , where $\hat{\gamma}=\max_{u,v}|V_g(f)(u,v)|$ is the maximum modulus of the STFT across the time-frequency plane, and $\hat{\sigma}$ is the standard deviation of the noise, estimated from the signal using a median absolute deviation estimator \cite{donoho1994ideal, mallat2008wavelet} \[ \hat{\sigma} = \frac{\text{median}\left(|\mathfrak{R}\{ V_{g}(f) \}|\right)}{0.6745}. \] 
The NSM serves as a practical proxy for the SNR; signals with higher NSM values are expected to have higher SNR values.
We then consider the set of GW signals as a union of the subsets $Z_i$ for $i \in \{2, 3, 4, 5, 6, 7\}$, where each $Z_i$ is defined to contain the signals with $\hat{\zeta} \ge i$. Note that $Z_i\subset Z_j$ if $i\ge j$.

We collected 50,000 signal samples for $Z_2$, and computed $\hat{\zeta}$ for each one, using the signal corresponding to LIGO Hanford. From this group, $|Z_3|=49,995$, $|Z_4|=14,542$, $|Z_5|=1,528$, $|Z_6|=670$, and $|Z_7|=255$.
Examples of spectrograms of the signals in the different subsets can be found in Sec. \ref{APP:s:gw_detection} of the SI document.
As can be noted in Figure \ref{f:subset_c}, most signals in $Z_7$ are visually identifiable in their spectrograms, while $Z_5$ and $Z_6$ contain numerous more subtle signals.
Note that our NSM thresholding is only used to create subsets of the data for comparison purposes; it plays no role in the detection procedure whatsoever.

To compute our ROC curves, for each value of $\hat{\zeta}$, we randomly selected an equal number of noise-only and signal-plus-noise samples in each subset of the data.
This makes sure we have class-balanced datasets for each value of $\hat{\zeta}$.

Finally, note that the signal components in GW signals typically span a narrow range of frequencies for more than half of the signal duration, causing them to persist for shorter periods in the persistence tree.
Hence, energy-based test statistics are expected to enhance detection performance, compared to distance-based tests.
Detection was carried out by applying each test to the three interferometers and claiming a detection when the null was rejected for at least one of the three.

\begin{figure}
    \centering
    \begin{subfigure}[b]{0.98\textwidth}
        \centering  
    \includegraphics{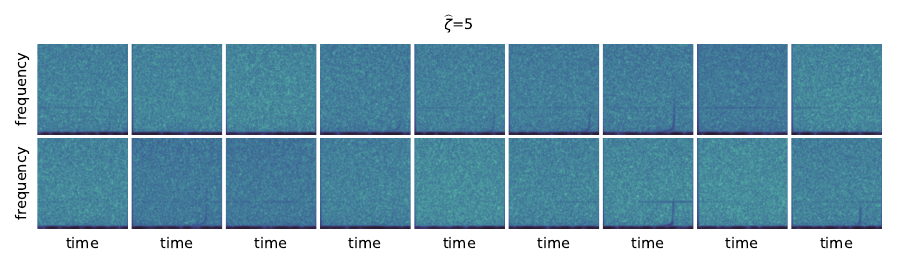} \\
    \includegraphics{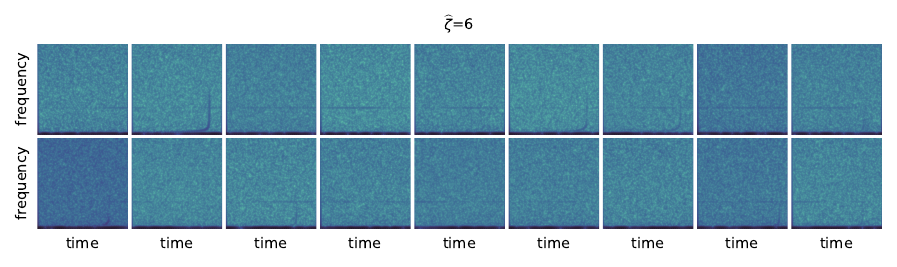} \\
    \includegraphics{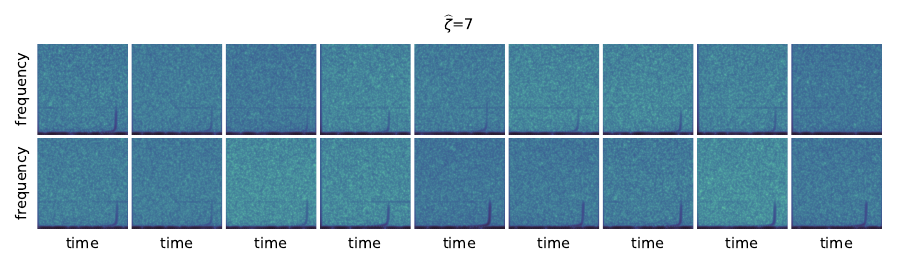}
    \caption{Random subsets of spectrograms from the GW dataset with thresholds $\hat{\zeta}\in\{5,6,7\}$. The signals correspond to the LIGO Hanford observatory.} \label{f:subset_c}  
    \end{subfigure}
    
    \interspace
    
    \begin{subfigure}[b]{0.31\textwidth}
        \hspace*{-1.5cm}
    \includegraphics{ROC_energy_bonferroni_rescaled.pdf}
    \caption{} \label{f:ROC_c}   
    \end{subfigure}
    \caption{Gravitational wave (GW) detection. \textbf{(a)} shows spectrograms of random subsets of signals coming from the dataset with thresholds $\hat{\zeta}\in{5,6,7}$. All the signals have a GW present, even though is not evident in most of the signals from the sets with $\hat{\zeta}=5$ and $\hat{\zeta}=6$ \textbf{(b)} Gravitational wave detection ROC curves. Test statistic: energy (Bonferroni).}
\end{figure}

\newpage
\bibliographystyle{unsrt}

\begin{thebibliography}{10}

\bibitem{NMTFYOHMK24}
R.~Nagaoka, K.~Masuzawa, M.~Taniwaki, A.~L. Foggiatto, T.~Yamazaki, I.~Obayashi, Y.~Hiraoka, C.~Mitsumata, and M.~Kotsugi.
\newblock Quantification of the coercivity factor in soft magnetic materials at different frequencies using topological data analysis.
\newblock {\em IEEE Transactions on Magnetics}, 2024.

\bibitem{CiHiVi23}
A.~Cipriani, C.~Hirsch, and M.~Vittorietti.
\newblock Topology-based goodness-of-fit tests for sliced spatial data.
\newblock {\em Computational Statistics \& Data Analysis}, 179:107655, 2023.

\bibitem{NiLeCa11}
M.~Nicolau, A.~J. Levine, and G.~Carlsson.
\newblock Topology-based data analysis identifies a subgroup of breast cancers with a unique mutational profile and excellent survival.
\newblock {\em Proceedings of the National Academy of Sciences}, 108(17):7265--7270, 2011.

\bibitem{xia2014persistent}
K.~Xia and G.-W. Wei.
\newblock Persistent homology analysis of protein structure, flexibility, and folding.
\newblock {\em International journal for numerical methods in biomedical engineering}, 30(8):814--844, 2014.

\bibitem{RFDHB24}
W.~Reise, X.~Fern{\'a}ndez, M.~Dominguez, H.~A. Harrington, and M.~Beguerisse-Diaz.
\newblock Topological fingerprints for audio identification.
\newblock {\em SIAM Journal on Mathematics of Data Science}, 6(3):815--841, 2024.

\bibitem{iqbal2021classification}
Sohail Iqbal, H~Fareed Ahmed, Talha Qaiser, Muhammad~Imran Qureshi, and Nasir Rajpoot.
\newblock Classification of {COVID-19} via homology of {CT-SCAN}.
\newblock {\em arXiv preprint arXiv:2102.10593}, 2021.

\bibitem{colombo2022tool}
Gloria Colombo, Ryan John~A Cubero, Lida Kanari, Alessandro Venturino, Rouven Schulz, Martina Scolamiero, Jens Agerberg, Hansruedi Mathys, Li-Huei Tsai, Wojciech Chach{\'o}lski, et~al.
\newblock A tool for mapping microglial morphology, {morphOMICs}, reveals brain-region and sex-dependent phenotypes.
\newblock {\em Nature neuroscience}, 25(10):1379--1393, 2022.

\bibitem{biscio2019accumulated}
Christophe~AN Biscio and Jesper M{\o}ller.
\newblock The accumulated persistence function, a new useful functional summary statistic for topological data analysis, with a view to brain artery trees and spatial point process applications.
\newblock {\em Journal of Computational and Graphical Statistics}, 28(3):671--681, 2019.

\bibitem{byrne2019topological}
Helen~M Byrne, Heather~A Harrington, Ruth Muschel, Gesine Reinert, Bernadette~J Stolz, and Ulrike Tillmann.
\newblock Topological methods for characterising spatial networks: A case study in tumour vasculature, 2019.

\bibitem{kerber2016persistent}
Michael Kerber.
\newblock Persistent homology: state of the art and challenges.
\newblock {\em International Mathematische Nachrichten}, 231(15-33):1, 2016.

\bibitem{carlsson2020topological}
Gunnar Carlsson.
\newblock Topological methods for data modelling.
\newblock {\em Nature Reviews Physics}, 2(12):697--708, 2020.

\bibitem{bobrowski2023universal}
Omer Bobrowski and Primoz Skraba.
\newblock A universal null-distribution for topological data analysis.
\newblock {\em Scientific reports}, 13(1):12274, 2023.

\bibitem{ligo2020guide}
LIGO~Scientific Collaboration, Virgo Collaboration, et~al.
\newblock A guide to {LIGO}--{Virgo} detector noise and extraction of transient gravitational-wave signals.
\newblock {\em Preprint}, 2020.

\bibitem{zeng23}
Zhen Zeng, Rachneet Kaur, Suchetha Siddagangappa, Tucker Balch, and Manuela Veloso.
\newblock From pixels to predictions: Spectrogram and vision transformer for better time series forecasting.
\newblock In {\em Proceedings of the Fourth ACM International Conference on AI in Finance}, ICAIF '23, page 82–90, New York, NY, USA, 2023. Association for Computing Machinery.

\bibitem{xie2024}
Luyuan Xie.
\newblock {TRLS}: A time series representation learning framework via spectrogram for medical signal processing.
\newblock In {\em 2024 IEEE International Conference on Acoustics, Speech and Signal Processing (ICASSP)}. IEEE, IEEE Signal Processing Society SigPort, 2024.

\bibitem{du2020imageprocessingtoolsfinancial}
Bairui Du, Delmiro Fernandez-Reyes, and Paolo Barucca.
\newblock Image processing tools for financial time series classification, 2020.

\bibitem{wu2020current}
Hau-Tieng Wu.
\newblock Current state of nonlinear-type time--frequency analysis and applications to high-frequency biomedical signals.
\newblock {\em Current Opinion in Systems Biology}, 23:8--21, 2020.

\bibitem{Fla}
Patrick Flandrin.
\newblock Time--frequency filtering based on spectrogram zeros.
\newblock {\em IEEE Signal Processing Letters}, 22(11):2137--2141, 2015.

\bibitem{meignen2016adaptive}
Sylvain Meignen, Thomas Oberlin, Philippe Depalle, Patrick Flandrin, and Stephen McLaughlin.
\newblock Adaptive multimode signal reconstruction from time--frequency representations.
\newblock {\em Philosophical Transactions of the Royal Society A: Mathematical, Physical and Engineering Sciences}, 374(2065):20150205, 2016.

\bibitem{BaFlCh18}
R{\'e}mi Bardenet, Julien Flamant, and Pierre Chainais.
\newblock On the zeros of the spectrogram of white noise.
\newblock {\em Applied and Computational Harmonic Analysis}, 48(2):682--705, 2020.

\bibitem{BaHa19}
R{\'e}mi Bardenet and Adrien Hardy.
\newblock Time-frequency transforms of white noises and {G}aussian analytic functions.
\newblock {\em Applied and computational harmonic analysis}, 50:73--104, 2021.

\bibitem{koliander2019filtering}
G{\"u}nther Koliander, Luis~Daniel Abreu, Antti Haimi, and Jos{\'e}~Luis Romero.
\newblock Filtering the continuous wavelet transform using hyperbolic triangulations.
\newblock In {\em 2019 13th International conference on Sampling Theory and Applications (SampTA)}, pages 1--4. IEEE, 2019.

\bibitem{escudero2024efficient}
Luis~Alberto Escudero, Naomi Feldheim, G{\"u}nther Koliander, and Jos{\'e}~Luis Romero.
\newblock Efficient computation of the zeros of the {B}argmann transform under additive white noise.
\newblock {\em Foundations of Computational Mathematics}, 24(1):279--312, 2024.

\bibitem{ghosh2022signal}
Subhroshekhar Ghosh, Meixia Lin, and Dongfang Sun.
\newblock Signal analysis via the stochastic geometry of spectrogram level sets.
\newblock {\em IEEE Transactions on Signal Processing}, 70:1104--1117, 2022.

\bibitem{haimi2022zeros}
Antti Haimi, G{\"u}nther Koliander, and Jos{\'e}~Luis Romero.
\newblock Zeros of {G}aussian {W}eyl--{H}eisenberg functions and hyperuniformity of charge.
\newblock {\em Journal of Statistical Physics}, 187(3):22, 2022.

\bibitem{miramont2023unsupervised}
Juan~M Miramont, Fran{\c{c}}ois Auger, Marcelo~A Colominas, Nils Laurent, and Sylvain Meignen.
\newblock Unsupervised classification of the spectrogram zeros with an application to signal detection and denoising.
\newblock {\em Signal Processing}, page 109250, 2023.

\bibitem{moukadem2024analytic}
Ali Moukadem, Barbara Pascal, Jean-Baptiste Courbot, and Nicolas Juillet.
\newblock The analytic {S}tockwell transform and its zeros.
\newblock {\em arXiv preprint arXiv:2407.17076}, 2024.

\bibitem{PaBa24Sub}
Barbara Pascal and R{\'e}mi Bardenet.
\newblock Point processes and spatial statistics in time-frequency analysis.
\newblock {\em arXiv preprint arXiv:2402.19172}, 2024.

\bibitem{gardner2006sparse}
Timothy~J Gardner and Marcelo~O Magnasco.
\newblock Sparse time-frequency representations.
\newblock {\em Proceedings of the National Academy of Sciences}, 103(16):6094--6099, 2006.

\bibitem{PaBa22}
B.~Pascal and R.~Bardenet.
\newblock A covariant, discrete time-frequency representation tailored for zero-based signal detection.
\newblock {\em IEEE Transactions on Signal Processing}, 2022.

\bibitem{miramont2024benchmarking}
Juan~M. Miramont, Rémi Bardenet, Pierre Chainais, and Francois Auger.
\newblock Benchmarking multi-component signal processing methods in the time-frequency plane, 2024.

\bibitem{hough2009zeros}
John~Ben Hough, Manjunath Krishnapur, Yuval Peres, et~al.
\newblock {\em Zeros of Gaussian analytic functions and determinantal point processes}, volume~51.
\newblock American Mathematical Soc., 2009.

\bibitem{torquato2016hyperuniformity}
Salvatore Torquato.
\newblock Hyperuniformity and its generalizations.
\newblock {\em Phys. Rev. E}, 94:022122, Aug 2016.

\bibitem{abreu2017weyl}
Luis~Daniel Abreu, Joao~M Pereira, Jos{\'e}~Luis Romero, and Salvatore Torquato.
\newblock The {W}eyl--{H}eisenberg ensemble: hyperuniformity and higher {L}andau levels.
\newblock {\em Journal of Statistical Mechanics: Theory and Experiment}, 2017(4):043103, 2017.

\bibitem{salvalaglioAPS}
Marco Salvalaglio, Dominic~J. Skinner, J\"orn Dunkel, and Axel Voigt.
\newblock Persistent homology and topological statistics of hyperuniform point clouds.
\newblock {\em Phys. Rev. Res.}, 6:023107, May 2024.

\bibitem{torquato2023hyperuniformity}
Charles~Emmett Maher, Yang Jiao, and Salvatore Torquato.
\newblock Hyperuniformity of maximally random jammed packings of hyperspheres across spatial dimensions.
\newblock {\em Phys. Rev. E}, 108:064602, Dec 2023.

\bibitem{Milor_2025}
Abel H~G Milor and Marco Salvalaglio.
\newblock Inferring traits of hyperuniformity from local structures via persistent homology.
\newblock {\em Journal of Physics: Condensed Matter}, 37(14):145401, feb 2025.

\bibitem{flandrin1998time}
Patrick Flandrin.
\newblock {\em Time-frequency/time-scale analysis}, volume~10.
\newblock Academic press, 1998.

\bibitem{grochenig2001foundations}
Karlheinz Gr{\"o}chenig.
\newblock {\em Foundations of time-frequency analysis}.
\newblock Springer Science \& Business Media, 2001.

\bibitem{flandrin2018explorations}
Patrick Flandrin.
\newblock {\em Explorations in time-frequency analysis}.
\newblock Cambridge University Press, 2018.

\bibitem{AsBr09}
G.~Ascensi and J.~Bruna.
\newblock Model space results for the {G}abor and {W}avelet transforms.
\newblock {\em IEEE Transactions on Information Theory}, 55(5):2250--2259, 2009.

\bibitem{holden1996stochastic}
Helge Holden, Bernt {\O}ksendal, Jan Ub{\o}e, and Tusheng Zhang.
\newblock Stochastic partial differential equations.
\newblock In {\em Stochastic partial differential equations}, pages 141--191. Springer, 1996.

\bibitem{chacholski2018building}
W.~Chach{\'o}lski, T.~Dyckerhoff, J.~Greenlees, G.~Stevenson, D.~Herbera, W.~Pitsch, and S.~Zarzuela.
\newblock {\em Building Bridges Between Algebra and Topology}.
\newblock Advanced Courses in Mathematics - CRM Barcelona. Springer International Publishing, 2018.

\bibitem{ChBe21}
F.~Chazal and B.~Michel.
\newblock An introduction to topological data analysis: fundamental and practical aspects for data scientists.
\newblock {\em Frontiers in artificial intelligence}, 4:667963, 2021.

\bibitem{EdHa22}
H.~Edelsbrunner and J.~L. Harer.
\newblock {\em Computational topology: an introduction}.
\newblock American Mathematical Society, 2022.

\bibitem{dey2022computational}
T.~K. Dey and Y.~Wang.
\newblock {\em Computational topology for data analysis}.
\newblock Cambridge University Press, 2022.

\bibitem{Oba18}
I.~Obayashi.
\newblock {V}olume-optimal cycle: {T}ightest representative cycle of a generator in persistent homology.
\newblock {\em SIAM Journal on Applied Algebra and Geometry}, 2(4):508--534, 2018.

\bibitem{Oba23}
Ippei Obayashi.
\newblock Stable volumes for persistent homology.
\newblock {\em Journal of Applied and Computational Topology}, pages 1--36, 2023.

\bibitem{Sch15}
B.~Schweinhart.
\newblock {\em Statistical Topology of Embedded Graphs}.
\newblock PhD thesis, Princeton University, 2015.

\bibitem{cohen2005stability}
David Cohen-Steiner, Herbert Edelsbrunner, and John Harer.
\newblock Stability of persistence diagrams.
\newblock In {\em Proceedings of the twenty-first annual symposium on Computational geometry}, pages 263--271, 2005.

\bibitem{Ben20}
Paul Bendich, Peter Bubenik, and Alexander Wagner.
\newblock Stabilizing the unstable output of persistent homology computations.
\newblock {\em Journal of Applied and Computational Topology}, 4(2):309--338, 2020.

\bibitem{wasserman2013all}
Larry Wasserman.
\newblock {\em All of statistics: a concise course in statistical inference}.
\newblock Springer Science \& Business Media, 2013.

\bibitem{dudoit2008multiple}
Sandrine Dudoit, Mark~J Van Der~Laan, and Mark~J van~der Laan.
\newblock {\em Multiple testing procedures with applications to genomics}, volume~10.
\newblock Springer, 2008.

\bibitem{fromont2016family}
Magalie Fromont, Matthieu Lerasle, and Patricia Reynaud-Bouret.
\newblock Family-wise separation rates for multiple testing.
\newblock {\em The Annals of Statistics}, 44(6):2533--2563, 2016.

\bibitem{baddeley2014tests}
Adrian Baddeley, Peter~J Diggle, Andrew Hardegen, Thomas Lawrence, Robin~K Milne, and Gopalan Nair.
\newblock On tests of spatial pattern based on simulation envelopes.
\newblock {\em Ecological Monographs}, 84(3):477--489, 2014.

\bibitem{myllymaki2017global}
Mari Myllym{\"a}ki, Tom{\'a}{\v{s}} Mrkvi{\v{c}}ka, Pavel Grabarnik, Henri Seijo, and Ute Hahn.
\newblock Global envelope tests for spatial processes.
\newblock {\em Journal of the Royal Statistical Society: Series {B} (Statistical Methodology)}, 79(2):381--404, 2017.

\bibitem{gudhi:urm}
{The GUDHI Project}.
\newblock {\em {GUDHI} User and Reference Manual}.
\newblock {GUDHI Editorial Board}, 2015.

\bibitem{kaggle-g2net}
Chris Messenger, Christopher Zerafa, Elena Cuoco, Michael~J. Williams, and Walter Reade.
\newblock {G2N}et gravitational wave detection, \url{https://kaggle.com/competitions/g2net-gravitational-wave-detection}, 2021.

\bibitem{donoho1994ideal}
David~L Donoho and Jain~M Johnstone.
\newblock Ideal spatial adaptation by wavelet shrinkage.
\newblock {\em Biometrika}, 81(3):425--455, 1994.

\bibitem{mallat2008wavelet}
Stephane Mallat.
\newblock {\em A {W}avelet {T}our of {S}ignal {P}rocessing: {T}he {S}parse {W}ay}.
\newblock Academic Press, 2008.

\bibitem{pham2018novel}
Duong-Hung Pham and Sylvain Meignen.
\newblock A novel thresholding technique for the denoising of multicomponent signals.
\newblock In {\em 2018 IEEE International Conference on Acoustics, Speech and Signal Processing (ICASSP)}, pages 4004--4008. IEEE, 2018.

\bibitem{chassande1999time}
Eric Chassande-Mottin and Patrick Flandrin.
\newblock On the time--frequency detection of chirps.
\newblock {\em Applied and Computational Harmonic Analysis}, 6(2):252--281, 1999.

\bibitem{recommendation2001perceptual}
ITU-T Recommendation.
\newblock Perceptual evaluation of speech quality {(PESQ)}: An objective method for end-to-end speech quality assessment of narrow-band telephone networks and speech codecs.
\newblock {\em Rec. ITU-T P. 862}, 2001.

\bibitem{pesq_python}
Rafael G.~Dantas Miao~Wang, Christoph~Boeddeker and Ananda Seelan.
\newblock Pesq (perceptual evaluation of speech quality) wrapper for python users, May 2022.

\bibitem{torcoli2021objective}
Matteo Torcoli, Thorsten Kastner, and J{\"u}rgen Herre.
\newblock Objective measures of perceptual audio quality reviewed: An evaluation of their application domain dependence.
\newblock {\em IEEE/ACM Transactions on Audio, Speech, and Language Processing}, 29:1530--1541, 2021.

\bibitem{andrzejak2001indications}
Ralph~G Andrzejak, Klaus Lehnertz, Florian Mormann, Christoph Rieke, Peter David, and Christian~E Elger.
\newblock Indications of nonlinear deterministic and finite-dimensional structures in time series of brain electrical activity: Dependence on recording region and brain state.
\newblock {\em Physical Review E}, 64(6):061907, 2001.

\bibitem{guillemard2011groupoid}
Mijail Guillemard and Armin Iske.
\newblock On groupoid {C}*-algebras, persistent homology and time-frequency analysis.
\newblock {\em Preprint}, 105, 2011.

\bibitem{boche2013signal}
Holger Boche, Mijail Guillemard, Gitta Kutyniok, and Friedrich Philipp.
\newblock Signal analysis with frame theory and persistent homology.
\newblock In {\em The Conference of Sampling Theory and Applications, SampTA’13}, volume~1, page~1, 2013.

\bibitem{whalen2013detection}
Anthony~D Whalen.
\newblock {\em Detection of signals in noise}.
\newblock Academic press, 2013.

\bibitem{ashok2021new}
Anjana Ashok, Banafsheh Beheshtipour, Maria~Alessandra Papa, Paulo~CC Freire, Benjamin Steltner, Bernd Machenschalk, Oliver Behnke, Bruce Allen, and Reinhard Prix.
\newblock New searches for continuous gravitational waves from seven fast pulsars.
\newblock {\em The Astrophysical Journal}, 923(1):85, 2021.

\end{thebibliography}

\end{document}